\documentclass[envcountsame,runningheads]{llncs}
\usepackage{graphicx}
\usepackage{amsmath}
\usepackage{amssymb}
\usepackage[unicode=true]{hyperref}
\newcommand\rurl[1]{\href{http://#1}{\nolinkurl{#1}}}
\usepackage{ifluatex}
\ifluatex
  \usepackage{fontspec}
  \defaultfontfeatures{Ligatures=TeX}
  \usepackage[po,splitting=weightedMedian]{luatodonotes}
\else
  \usepackage[utf8]{inputenc}
  \usepackage{todonotes}
\fi

\usepackage{enumerate}
\usepackage{xspace}
\usepackage{booktabs}
\usepackage{tikz}
\usetikzlibrary{calc}
\usepackage[font=small,labelfont=bf,labelsep=period]{caption}
\usepackage[font=small,labelfont=normalfont]{subcaption}
\usepackage{wrapfig}
\usepackage{timetravel}
\setCounters{theorem}
\usepackage{arxiv-appendix}
\arxivid{1607.01196v2}
\arxivcite{ChaplickFLRVW16}
\arxivauxname{rho-geometry-arxiv}
\lncsfalse

\title{Drawing Graphs on Few Lines and Few Planes%
\iflncs%
  \thanks{\arxivrefthanks}
\else%
  \thanks{Appears in the Proceedings of the 24th International Symposium on
  Graph Drawing and Network Visualization (GD 2016).}
\fi%
}

\author{Steven Chaplick \inst1 \and Krzysztof
Fleszar \inst1 \and Fabian Lipp \inst1 \and Alexander~Ravsky \inst2
\and Oleg~Verbitsky \inst3 \and Alexander~Wolff \inst1}

\institute{Lehrstuhl f\"ur Informatik~I, Universit\"at W\"urzburg, Germany,
  \rurl{www1.informatik.uni-wuerzburg.de/en/staff}\and
Pidstryhach Institute for Applied Problems of Mechanics and
Mathematics, National Academy
of Science of Ukraine, Lviv, Ukraine \and
Institut f\"ur Informatik,
Humboldt-Universit\"at zu Berlin, Germany}

\authorrunning{S.~Chaplick et al.}

\newcommand{\of}[1]{\left( #1 \right)}
\newcommand{\function}[2]{:#1 \rightarrow #2}
\newcommand{\setdef}[2]{\left\{ \hspace{0.5mm} #1 : \hspace{0.5mm} #2 \right\}}
\newcommand{\reals}{\mathbb{R}}
\newcommand{\maxdeg}{\Delta}
\DeclareMathOperator{\diam}{diam}
\DeclareMathOperator{\tn}{tn}
\DeclareMathOperator{\tw}{tw}
\DeclareMathOperator{\bw}{bw}
\DeclareMathOperator{\sep}{sep}
\DeclareMathOperator{\Conv}{Conv}
\DeclareMathOperator{\es}{es}
\DeclareMathOperator{\la}{la}
\DeclareMathOperator{\segm}{segm}
\DeclareMathOperator{\slop}{slop}
\DeclareMathOperator{\area}{area}
\DeclareMathOperator{\vt}{vt}
\DeclareMathOperator{\lvaname}{lva}
\newcommand{\lva}[1]{\lvaname(#1)}
\newcommand{\Wseparator}{$W$\kern-0.12em -separator\xspace}

\usepackage{comment}
\excludecomment{question}

\let\doproof\proof
\def\proof{\gdef\myqedtmp{\qed}\doproof}
\let\doendproof\endproof
\renewcommand\endproof{\myqedtmp\doendproof}
\newcommand{\qedhere}{\myqedtmp\gdef\myqedtmp{}}

\graphicspath{{figures/}}

\newcommand{\savespace}[1]{}

\begin{document}

\maketitle

\begin{abstract}
  We investigate the problem of drawing graphs in 2D and 3D such that
  their edges (or only their vertices) can be covered by few
  lines or planes.  We insist on straight-line edges and crossing-free
  drawings.  This problem has many connections to other challenging
  graph-drawing problems such as small-area or small-volume drawings,
  layered or track drawings,
and drawing graphs
  with low visual complexity.  While some facts about our problem are
  implicit in previous work, this is the first treatment of the
  problem in its full generality.
Our contribution is as follows.
\begin{itemize}
\item We show lower and upper bounds for the numbers of lines and
  planes needed for covering drawings of graphs in certain graph
  classes.  In some cases our bounds are asymptotically tight; in some cases we are
  able to determine exact values.
\item We relate our parameters to standard combinatorial
  characteristics of graphs (such as the chromatic number, treewidth,
  maximum degree, or arboricity) and to parameters that have been
  studied in graph drawing (such as the track number or the number of
  segments appearing in a drawing).
\item We pay special attention to planar graphs.  For example, we show
  that there are planar graphs that can be drawn in 3-space on a lot
  fewer lines than in the plane.
\end{itemize}
\end{abstract}

\section{Introduction}

It is well known that any graph admits a straight-line drawing in 3-space.
Suppose that we are allowed to draw edges only on a limited number of
planes. How many planes do we need for a given graph $G$?
For example, $K_6$ needs four planes; see Fig.~\ref{fig:K6-photo}.
Note that this question is different from the well-known concept
of a \emph{book embedding} where all vertices lie on one line (the
spine) and edges lie on a limited number of adjacent half-planes
(the pages).  In contrast, we put no restriction
on the mutual position of planes, the vertices can be located
in the planes arbitrarily, and the edges must be straight-line.

\begin{figure}
\begin{minipage}[b]{.57\textwidth}
  \centering
  \includegraphics{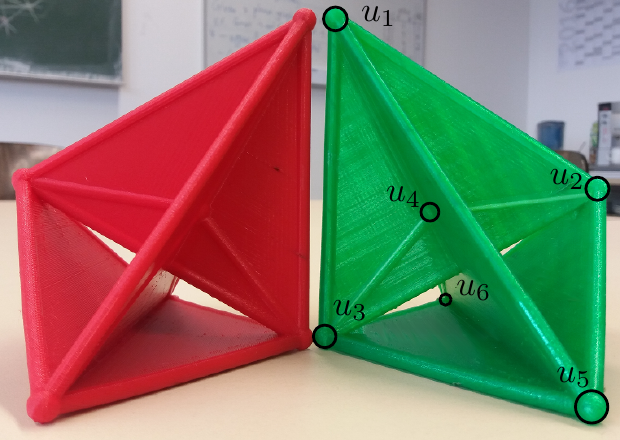}
  \caption{$K_6$ can be drawn straight-line and crossing-free on four
    planes.  This is optimal, that is, $\rho^2_3(K_6)=4$.}
  \label{fig:K6-photo}
\end{minipage}
\hfill
\begin{minipage}[b]{.4\textwidth}
  \includegraphics[page=1,scale=1.2]{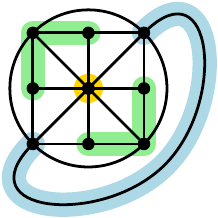}
  \vspace*{-6ex}

  \hspace{0ex}\hfill\includegraphics[page=3,scale=1.2]{pi13is3}
  \caption{Planar 9-vertex graph~$G$ with $\pi^1_3(G)=3$,
    3D-drawing on three lines.}
  \label{fig:three-lines}
\end{minipage}
\end{figure}
In a weaker setting, we require only the vertices to be located
on a limited number of planes (or lines).  For example, the graph in
Fig.~\ref{fig:three-lines}
can be drawn in 2D such that its vertices
are contained in three lines; we conjecture that it is
the smallest planar graph that needs more than two lines even in 3D.
This version of our problem is related to the well-studied
problem of drawing a graph straight-line in a 3D grid
of bounded volume \cite{DujmovicW13,w-3dgd-EA08}:
If a graph can be drawn with all vertices on a
grid of volume~$v$, then $v^{1/3}$ planes and $v^{2/3}$ lines suffice.
We now formalize the problem.

\begin{definition}\rm
  Let $1 \le l < d$, and let $G$ be a graph.  We define
  the \emph{$l$-dimensional affine cover number} of~$G$ in~$\reals^d$,
  denoted by $\rho^l_d(G)$, as the minimum number of $l$-dimensional planes
  in~$\mathbb R^d$ such that~$G$ has a drawing that is contained in
  the union of these planes.  We define $\pi^l_d(G)$, the
  \emph{weak} $l$-dimensional affine cover number of~$G$ in~$\reals^d$,
   similarly to $\rho^l_d(G)$,
  but under the weaker restriction that the vertices (and not
  necessarily the edges) of~$G$ are contained in the union of the
  planes.  Finally, the \emph{parallel} affine cover number, $\bar\pi^l_d(G)$,
  is a restricted version
  of~$\pi^l_d(G)$, in which we insist that the planes are parallel.  We
  consider only straight-line and crossing-free drawings. Note: $\rho^l_d(G)$,
  $\pi^l_d(G)$, and $\bar\pi^l_d(G)$ are only undefined when $d=2$ and $G$ is
  non-planar.
\end{definition}

Clearly, for any combination of $l$ and $d$, it holds that
$\pi^l_d(G)\le\bar\pi^l_d(G)$ and $\pi^l_d(G)\le\rho^l_d(G)$.
Larger values of~$l$ and~$d$ give us more freedom for drawing graphs
and, therefore, smaller $\pi$- and $\rho$-values.  Formally, for any
graph~$G$, if $l'\le l$ and $d'\le d$
then~$\pi^l_d(G)\le \pi^{l'}_{d'}(G)$, $\rho^l_d(G)\le \rho^{l'}_{d'}(G)$, and
$\bar\pi^l_d(G)\le \bar\pi^{l'}_{d'}(G)$.
But in most cases this freedom is not essential.
For example, it suffices to consider $l\le2$ because otherwise
$\rho^l_d(G)=1$.  More interestingly, we can actually focus on $d\le3$
because every graph can be drawn in~$3$-space as effectively as in high dimensional spaces, i.e.,
  for any integers $1 \le l \le d$, $d\ge 3$, and for any graph $G$,
  it holds that $\pi^l_d(G)=\pi^l_3(G)$,
  $\bar\pi^l_d(G)=\bar\pi^l_3(G)$, and $\rho^l_d(G)=\rho^l_3(G)$.
We prove this important fact in
Appendix~\appref{s:collapse}.  Thus, our task is to investigate the cases
$1\le l<d\le 3$.
We call $\rho^1_2(G)$ and $\rho^1_3(G)$ the \emph{line cover numbers} in 2D and
3D, $\rho^2_3(G)$ the \emph{plane cover number}, and analogously for the weak
versions.

\paragraph{Related work.}

We have already briefly mentioned 3D graph drawing on the grid, which
has been surveyed by Wood~\cite{w-3dgd-EA08} and by Dujmovi{\'c} and
Whitesides \cite{DujmovicW13}.
For example, Dujmovi\'c~\cite{d-glvls-JCTSB15}, improving on a result
of Di Battista et al.~\cite{BattistaFP13}, showed that any planar
graph can be drawn into a 3D-grid of volume $O(n \log n)$.
It is well-known that, in 2D, any planar graph admits a plane straight-line
drawing on an $O(n) \times O(n)$ grid~\cite{s-epgg-SODA90,fpp-hdpgg-Comb90}
and that the nested-triangles graph~$T_k = K_3
\times P_k$ (see Fig.~\ref{fig:nested-triangles}, p.~\pageref{fig:nested-triangles}) with $3k$ vertices
needs $\Omega(k^2)$ area~\cite{fpp-hdpgg-Comb90}.

An interesting variant of our problem is to study drawings whose edge
sets are represented (or covered) by as few objects as possible.  The
type of objects that have been used are straight-line segments
\cite{desw-dpgfs-CGTA07,dm-dptfs-CCCG14} and circular arcs
\cite{s-dgfa-JGAA15}.  The idea behind this objective is to keep the
visual complexity of a drawing low for the observer.  For example,
Schulz~\cite{s-dgfa-JGAA15} showed how to draw the dodecahedron
by using~$10$ arcs, which is optimal.

\savespace{
Classes of planar graphs with bounded $\pi^1_2$-value are of interest
in the context of \emph{untang\-ling}~\cite{pt-up-02}.
If $\pi^1_2(G) \le c$, then $G$ admits a straight-line drawing with at
least $n/c$ collinear vertices.
If~$G$ admits a set of~$k$ collinear vertices with a special property,
then any straight-line drawing of~$G$ can be untangled while at
least~$\sqrt k$ vertices remain fixed~\cite{RavskyV11}.
}

\paragraph{Our contribution.}
Our research goes into three directions.

First, we show lower and upper bounds
for the numbers of lines and planes needed for covering drawings of
graphs in certain graph classes such as
graphs of bounded degree or subclasses of planar graphs.
The most natural graph families to start with are
the complete graphs and the complete bipartite graphs.
Most versions of the affine cover numbers of these graphs
can be determined easily. Two cases are much more subtle:
We determine $\rho^2_3(K_n)$ and $\rho^1_3(K_{n,n})$ only
asymptotically, up to a factor of 2 (see Theorem~\ref{thm:KnLowerUpper} and Example~\ref{ex:rho12(K_n)}).
Some efforts are made to compute the exact values of $\rho^2_3(K_n)$
for small $n$ (see Theorem~\ref{thm:rhoCompleteGraphs}).
As another result in this direction, we prove that
$\rho^1_3(G)>n/5$ for almost all cubic graphs on $n$ vertices (Theorem~\ref{thm:rho^1_3-small-edge-separators}(b)).

Second, we relate the affine cover numbers to standard combinatorial characteristics of graphs
and to parameters that have been studied in graph drawing.
In Section~\ref{sec:vertices}, we characterize $\pi^1_3(G)$ and $\pi^2_3(G)$
in terms of the \emph{linear vertex arboricity} and the \emph{vertex thickness}, respectively.
This characterization implies that both $\pi^1_3(G)$ and $\pi^2_3(G)$
are linearly related to the chromatic number of the graph~$G$.
Along the way, we refine a result of Pach et al.~\cite{ptt-3dgdg-GD97}
concerning the volume of 3D grid drawings (Theorem~\ref{PiChar0}).
We also prove that any graph $G$ has balanced separators of size at most~$\rho^1_3(G)$
and conclude from this that $\rho^1_3(G)\ge\tw(G)/3$, where $\tw(G)$
denotes the treewidth of $G$ (Theorem~\ref{thm:rho^1_3-small-edge-separators}).
In Section~\ref{ss:rho12rho13}, we analyze the relationship
between~$\rho^1_2(G)$ and the segment number $\segm(G)$ of a graph,
which was introduced by Dujmovi\'c et al.~\cite{desw-dpgfs-CGTA07}.
We prove that $\segm(G)=O(\rho^1_2(G)^2)$ for any connected $G$ and show that this bound is optimal
(see Theorem~\ref{thm:segm-vs-rho12} and Example~\ref{ex:rho12-segm}).

Third, we pay special attention to planar graphs (Section~\ref{s:planar}).
Among other results, we show examples of planar graphs with a large gap
between the parameters $\rho^1_3(G)$ and $\rho^1_2(G)$ (see Theorem~\ref{thm:nested-triangles}).

We also investigate the parallel affine cover numbers~$\bar\pi^1_2$ and~$\bar\pi^1_3$.
Observe that for any graph~$G$,
$\bar\pi^1_3(G)$ equals the \emph{improper track number} of~$G$, which
was introduced by Dujmovi{\'c} et al.~\cite{dmw-lgbtw-SICOMP05}.

Due to lack of space, our results for the parallel affine cover numbers
(along with a survey of known related results)
appear in Appendix~\appref{s:parallel}.
We defer some other proofs to
Appendices~\appref{sec:gen-appendix} and~\appref{sec:plan-appendix}
and list some open problems in Appendix~\appref{sec:open-problems}.

\paragraph{Remark on the computational complexity.}

In a follow-up paper~\cite{ChaplickFLRVW16b}, we investigate the computational complexity of
computing the $\rho$- and $\pi$-numbers.
We argue that it is NP-hard to decide whether a given graph has
a~$\pi^1_3$- or~$\pi^2_3$-value of~2 and that both values are
even hard to approximate. This result is based on Theorems~\ref{PiChar0}
and~\ref{thm:PiCharTwo} and Corollaries~\ref{PiChi} and~\ref{CorPiChar}
in the present paper.  While the graphs with $\rho_3^2$-value~1
are exactly the planar graphs (and hence, can be recognized in linear
time), it turns out that recognizing graphs with a
$\rho_3^2$-value of~2 is already NP-hard. In contrast to this,
the problems of deciding whether $\rho^1_3(G)\le k$ or $\rho^1_2(G)\le k$
are solvable in polynomial time for any fixed~$k$. However,
the versions of these problems with $k$ being part of the input are
complete for the complexity class $\exists\mathbb{R}$ which is based on
the \emph{existential theory of the reals} and
that plays an important role in computational geometry~\cite{SchaeferS15}.

\paragraph{Notation.}

For a graph $G=(V,E)$, we use $n$ and~$m$ to denote the numbers of
vertices and edges of~$G$, respectively.  Let $\maxdeg(G)=\max_{v \in V}
\deg(v)$ denote the maximum degree of~$G$. Furthermore, we will use
the standard notation
$\chi(G)$ for the chromatic number,
$\tw(G)$ for the treewidth,
and $\diam(G)$ for the diameter of~$G$.
The Cartesian product of graphs $G$ and $H$ is denoted by~$G\times H$.

\savespace{
\emph{Cubic} graphs are graphs where all vertices have degree~3.
Recall also that a graph is \emph{$k$-connected} if it has more than $k$ vertices and
stays connected after removal of any set of up to $(k-1)$ vertices.
A planar graph $G$ is \emph{maximal} if adding an edge between any two
non-adjacent vertices of $G$ violates planarity. Maximal planar graphs
on more than three vertices are
3-connected. Clearly, all facial cycles in such graphs have length 3.
By this reason maximal planar graphs are also called \emph{triangulations}.
}

\section{The Affine Cover Numbers in \texorpdfstring{$\reals^3$}{R³}}

\subsection{Placing Vertices on Few Lines or Planes
  (\texorpdfstring{$\pi^1_3$ and $\pi^2_3$}{π13 and π23})}
\label{sec:vertices}

A \emph{linear forest} is a forest whose connected components
are paths. The \emph{linear vertex arboricity} $\lva G$ of a graph~$G$
equals the smallest size $r$ of a partition $V(G)=V_1\cup\dots\cup V_r$
such that every $V_i$ induces a linear forest. This notion, which is an induced
version of the fruitful concept of \emph{linear arboricity} (see Remark~\ref{rem:rho13-vs-la} below),
appears very relevant to our topic.
The following result is based on a construction of Pach et al.~\cite{ptt-3dgdg-GD97};
see Appendix~\appref{sec:gen-appendix} for the proof.

\newcommand{\contentThmPiCharZero}{%
  For any graph~$G$, it holds that $\pi^1_3(G)=\lva G$. Moreover, any
  graph $G$ can be drawn with vertices on $r$ lines in the 3D integer grid of size
  $r\times 4rn\times 4r^2n$, where $r=\lva G$.
}
\wormhole{PiChar0}
\begin{theorem}
  \label{PiChar0}
  \contentThmPiCharZero
\end{theorem}

\begin{corollary}\label{PiChi}
$\chi(G)/2 \le\pi^1_3(G) \le \chi(G)$.
\end{corollary}

\savespace{
\begin{proof}
  We have $\lva G\le \chi(G)$ because any independent set is a linear
  forest.  On the other hand, $\chi(G)\le 2\lva G$ because any linear
  forest is 2-colorable.
\end{proof}
}

Corollary \ref{PiChi} readily implies that $\pi^1_3(G) \le
\Delta(G)+1$~\cite{b-ocnon-Cam41}.
This can be considerably improved using
a relationship between the linear vertex arboricity and the maximum degree
that is established by Matsumoto~\cite{Matsumoto90}.
Matsumoto's result implies that
$
\pi^1_3(G) \le \Delta(G)/2+1
$
for any connected graph~$G$.
Moreover, if $\Delta(G)=2d$, then $\pi^1_3(G)=d+1$ if and only if $G$ is
a cycle or the complete graph~$K_{2d+1}$.

We now turn to the weak plane cover numbers.
The \emph{vertex thickness} $\vt(G)$ of a graph $G$ is the smallest
size $r$ of a partition $V(G)=V_1\cup\dots\cup V_r$ such that
$G[V_1],\dots,G[V_r]$ are all planar.
We prove the following theorem in Appendix~\appref{sec:gen-appendix}.

\newcommand{\contentThmPiCharTwo}{%
  For any graph~$G$, it holds that $\pi^2_3(G)=\bar\pi^2_3(G)=\vt(G)$
  and that $G$ can be drawn such that all vertices lie on a 3D
  integer grid of size $\vt(G)\times O(m^2)\times O(m^2)$,
  where $m$ is the number of edges of $G$.
  Note that this drawing occupies $\vt(G)$ planes.

}
\wormhole{thm:PiCharTwo}
\begin{theorem}
  \label{thm:PiCharTwo}
  \contentThmPiCharTwo
\end{theorem}

\begin{corollary}\label{CorPiChar}
  $\chi(G)/4 \le \pi^2_3(G) \le \chi(G)$.
\end{corollary}

\newcommand{\contentExPiK}{%
  \begin{enumerate}[(a)]
  \item
\textit{
$\pi^1_3(K_n)=\lceil n/2\rceil$.}
\item \textit{
  $\pi^1_3(K_{p,q})=2$ for any $1 \le p \le q$; except for
  $\pi^1_3(K_{1,1})=\pi^1_3(K_{1,2})=1$.}
\item \textit{
$\pi^2_3(K_n)=\lceil n/4\rceil$; therefore,
$\pi^2_3(G)\le \lceil n/4\rceil$ for every graph~$G$.}
  \end{enumerate}%
}
\wormhole{pi^1_3(K_n)}
\begin{example}\label{pi^1_3(K_n)}
\contentExPiK
\end{example}

\subsection{Placing Edges on Few Lines or Planes
  (\texorpdfstring{$\rho^1_3$ and $\rho^2_3$}{ρ13 and ρ23})}
\label{sec:edges}

Clearly, $\maxdeg(G)/2\le\rho^1_3(G)\le m$ for any graph~$G$.
Call a vertex~$v$ of a graph~$G$ \emph{essential} if $\deg v\ge 3$ or
if~$v$ belongs to a~$K_3$ subgraph of~$G$.
Denote the number of essential vertices in~$G$ by~$\es(G)$.

\begin{lemma}\label{Ess}
  \begin{enumerate}[(a)]
  \item
$\rho^1_3(G) > (1+\sqrt{1+8\es(G)})/2$.%
\item
$\rho^1_3(G)>\sqrt{m^2/n-m}$ for any graph~$G$ with $m\ge n\ge 1$.
  \end{enumerate}
\end{lemma}

\begin{proof}
(a)
  In any drawing of a graph $G$, any essential vertex is shared
  by two edges not lying on the same line.  Therefore,
  each such vertex is an intersection point of at least two lines, which
  implies that $\es(G)\le \binom{\rho^1_3(G)}{2}$.  Hence,
  $\rho^1_3(G) \ge \big(1+\sqrt{1+8\es(G)}\,\big)/2>\sqrt{2\es(G)}$.

(b)
  Taking into account multiplicity of intersection points (that is,
  each vertex $v$ requires at least $\lceil\deg v/2\rceil (\lceil\deg
  v/2\rceil-1)/2$ intersecting line pairs), we obtain
  \begin{eqnarray*}
    \binom{\rho^1_3(G)}{2} & \ge & \frac 12 \sum_{v\in V(G)}
    \left\lceil\frac{\deg v}{2}
    \right\rceil\left(\left\lceil\frac{\deg v}2\right\rceil-1\right)
    \ge \sum \frac{\deg v(\deg v-2)}{8} = \\
    & = & \frac 18\sum (\deg v)^2-\frac 14\sum \deg v\ge \frac
    1{8n}\left(\sum \deg v\right)^2-\frac 14 2m
    =\frac{m^2}{2n}-\frac m2.
  \end{eqnarray*}
  The last inequality follows by the inequality between arithmetic and
  quadratic means.
  Hence, $\rho^1_3(G) > \sqrt{m^2/n-m}$.
\end{proof}

Part (a) of Lemma \ref{Ess} implies that $\rho^1_3(G)>\sqrt{2n}$
if a graph $G$ has no vertices of degree 1 and 2,
while Part (b) yields $\rho^1_3(G)>\sqrt{m/2}$ for all such~$G$.
Note that a disjoint union of $k$ cycles can have no essential vertices, but each
cycle will need $3$ intersection points of lines, i.e., such a graph has
$\rho^1_3 \in \Omega(\sqrt{k})$. Thus, $\rho^1_3$ cannot be bounded from above by a function
of essential vertices.

\begin{remark}\label{rem:rho13-vs-la}
The \emph{linear arboricity} $\la(G)$ of a graph $G$ is the minimum
number of linear forests which partition the edge set of~$G$; see~\cite{h-cagI-ANYAS70}.
Clearly, we have $\rho^1_3(G)\ge\la(G)$.
There is no function of $\la(G)$ that is an upper bound for $\rho^1_3(G)$.
  Indeed, let $G$ be an arbitrary cubic graph.  Akiyama et al.~\cite{AEH}
  showed that $\la(G)=2$.  On the other hand, any vertex of~$G$ is
  essential, so $\rho^1_3(G)>\sqrt{2n}$ by Lemma~\ref{Ess}(a). %
Theorem~\ref{thm:rho^1_3-small-edge-separators} below shows an even larger gap.
\end{remark}

We now prove a general lower bound for $\rho^1_3(G)$ in terms of the treewidth of $G$.
Note for comparison that $\pi^1_3(G) \le \chi(G) \le \tw(G)+1$ (the last inequality holds because the graphs of treewidth at most $k$ are exactly partial $k$-trees and the construction of a $k$-tree easily implies that it is $k+1$-vertex-chromatic)
\savespace{(and even $\bar\pi^1_3(G)$ is bounded from above by a function of $\tw(G)$;
see Section~\ref{barPi^1_3})}.
The relationship between $\rho^1_3(G)$ and $\tw(G)$ follows from the
fact that graphs with low parameter $\rho^1_3(G)$ have small separators.
This fact is interesting by itself and has yet another consequence:
Graphs with bounded vertex degree can have linearly large value of~$\rho^1_3(G)$
(hence, the factor of $n$ in the trivial bound $\rho^1_3(G)\le m\le\frac12\,n\,\Delta(G)$
is best possible).

We need the following definitions. Let $W\subseteq V(G)$.
A set of vertices $S\subset V(G)$ is a \emph{balanced \Wseparator} of
the graph $G$ if $|W\cap C|\le|W|/2$ for every
connected component $C$ of $G\setminus S$. Moreover, $S$ is a \emph{strongly balanced
\Wseparator} if there is a partition $W\setminus S=W_1\cup W_2$ such that
$|W_i|\le|W|/2$ for both $i=1,2$ and there is no path between $W_1$ and $W_2$ avoiding $S$.
Let $\sep_W(G)$ (resp.~$\sep^*_W(G)$) denote the minimum $k$ such that $G$ has
a (resp.\ strongly) balanced \Wseparator $S$ with $|S|=k$.
Furthermore, let $\sep(G)=\sep_{V(G)}(G)$ and $\sep^*(G)=\sep^*_{V(G)}(G)$.
Note that $\sep_W(G)\le\sep^*_W(G)$ for any $W$ and, in particular, $\sep(G)\le\sep^*(G)$.

It is known \cite[Theorem 11.17]{FlumG06} that $\sep_W(G)\le\tw(G)+1$ for every $W\subseteq V(G)$.
On the other hand, if $\sep_W(G)\le k$ for all $W$ with $|W|=2k+1$, then $\tw(G)\le3k$.

The \emph{bisection width} $\bw(G)$ of a graph $G$ is the minimum possible
number of edges between two sets of vertices $W_1$ and $W_2$ with $|W_1| =
\lceil n/2 \rceil$ and $|W_2| = \lfloor n/2 \rfloor$  partitioning $V(G)$.
Note that $\sep^*(G)\le\bw(G)+1$.

\begin{theorem}
  \label{thm:rho^1_3-small-edge-separators}
  \begin{enumerate}[(a)]
  \item

$\rho^1_3(G)\ge \bw(G)$.
\item
$\rho^1_3(G)>n/5$ for almost all cubic graphs with $n$ vertices.
\item
$\rho^1_3(G)\ge\sep^*_W(G)$ for every $W\subseteq V(G)$.
\item
$\rho^1_3(G) \ge \tw(G)/3$.
  \end{enumerate}
\end{theorem}
\begin{proof}
(a)
Fix a drawing of the graph $G$ on $r=\rho^1_3(G)$ lines in~$\mathbb R^3$.
Choose a plane $L$ that is not parallel
to any of the at most $\binom{n}{2}$ lines passing through two vertices of the drawing.
Let us move $L$ along the orthogonal direction
  until it separates the vertex set of~$G$ into two almost equal
  parts~$W_1$ and~$W_2$.
The plane~$L$ can intersect at
  most $r$ edges of~$G$, which implies that $\bw(G)\le r$.

(b) follows from Part (a) and the fact that a random cubic graph on $n$ vertices
has bisection width at least $n/4.95$ with probability $1-o(1)$
(Kostochka and Melnikov~\cite{KostochkaM93}).

(c) Given $W\subseteq V(G)$, we have to prove that $\sep^*_W(G)\le\rho^1_3(G)$.
Choose a plane $L$ as in the proof of Part (a) and move it
 until it separates $W$ into two equal parts~$W'_1$ and~$W'_2$;
if $|W|$ is odd, then $L$ should contain one vertex $w$ of $W$.
If $|W|$ is even, we can ensure that $L$ does not contain any vertex of $G$.
We now construct a set $S$ as follows.
If $L$ contains a vertex $w\in W$, i.e., $|W|$ is odd, we put $w$ in $S$.
Let $E$ be the set of those edges which are intersected by $L$ but
are not incident to the vertex $w$ (if it exists).
Note that $|E|<r$ if $|W|$ is odd and $|E|\le r$ if $|W|$ is even.
Each of the edges in $E$ contributes one of its incident vertices into $S$.
Note that $|S|\le r$.
Set $W_1=W'_1\setminus S$ and~$W_2=W'_2\setminus S$
and note that there is no edge between these sets of vertices.
Thus, $S$ is a strongly balanced \Wseparator.

(d) follows from (c) by the relationship between treewidth and
balanced separators.
\end{proof}

On the other hand, note that $\rho^1_3(G)$ cannot be bounded from above
by any function of $\tw(G)$. Indeed, by Lemma \ref{Ess}(a) we have $\rho^1_3(T)=\Omega(\sqrt n)$
for every caterpillar $T$ with linearly many vertices of degree 3.
The best possible relation in this direction is
$\rho^1_3(G) \le m < n \tw(G)$.
The factor $n$ cannot be improved here (take $G=K_n$).

\newcommand{\contentExRho}{%
  \begin{enumerate}[(a)]
  \item \textit{
$\rho^1_3(K_n)=\binom{n}{2}$ for any $n\ge 2$.}
  \item \textit{
$pq/2 \le \rho^1_3(K_{p,q}) \le pq$ for any $1 \le p \le q$.}
\end{enumerate}%
}
\wormhole{ex:rho12(K_n)}
\begin{example}\label{ex:rho12(K_n)}
\contentExRho
\end{example}

\begin{question} Can we tighten the bounds for $\rho^1_3(K_{p,q})$?
\end{question}

We now turn to the plane cover number.

\newcommand{\contentExRhoKpq}{%
\textit{For any integers $1 \le p \le q$, it holds that $\rho^2_3(K_{p,q}) = \lceil p/2 \rceil$.}%
}
\wormhole{Kpq}
\begin{example}\label{Kpq}
\contentExRhoKpq
\end{example}

\savespace{Since any planar graph has less than $3n$ edges, we easily
conclude that $\rho^1_3(G)<3n\rho^2_3(G)$. Examples~\ref{ex:rho12(K_n)}
and~\ref{Kpq} show that the factor of $n$ is best possible here.}

Determining the parameter $\rho^2_3(G)$ for complete graphs $G=K_n$
is a much more subtle issue. We are able to determine the asymptotics
of $\rho^2_3(K_n)$ up to a factor of~2.

By a \emph{combinatorial cover} of a graph $G$ we mean a set of subgraphs
$\{G_i\}$ such that every edge of $G$ belongs to $G_i$ for some~$i$.
A \emph{geometric cover} of a crossing-free drawing $d \colon V(K_n)\to\mathbb R^3$ of a
complete graph $K_n$ is a set $\mathcal L$ of planes in $\mathbb R^3$ so that for
each pair of vertices $v_i, v_j \in V(K_n)$ there is a plane $\ell \in
\mathcal L$ containing both points $d(v_i)$ and $d(v_j)$.
This geometric cover $\mathcal L$ induces a combinatorial cover
$\mathcal K_{\mathcal L} = \{G_\ell \mid \ell \in \mathcal L\}$ of the graph
$K_n$, where $G_\ell$ is the subgraph of $K_n$ induced by the set $d^{-1}(\ell)$.
Note that each $G_\ell$ is a $K_s$ subgraph with $s\le4$ (because $K_5$ is not planar).

Let $c(K_n,K_s)$ denote the minimum size of a combinatorial
cover of $K_n$ by $K_s$ subgraphs ($c(K_n,K_s)=0$ if $s>n$).
The asymptotics of the numbers $c(K_n,K_s)$ for $s=3,4$
can be determined via the results about \emph{Steiner systems}
by Kirkman and Hanani~\cite{Bol1,h-ecbibd-AMS61}.
This yields the following bounds for~$\rho^2_3(K_n)$
(see Appendix~\appref{sec:gen-appendix}).

\newcommand{\contentThmKnLowerUpper}{%
For all~$n\ge3$,
$$
\of{1/2+o(1)}n^2 = c(K_n,K_4) \le \rho^2_3(K_n) \le c(K_n,K_3) = \of{1/6+o(1)} n^2.
$$%
}
\wormhole{thm:KnLowerUpper}
\begin{theorem}
  \label{thm:KnLowerUpper}
  \contentThmKnLowerUpper
\end{theorem}

Note that we cannot always realize a combinatorial cover of $K_n$ by
copies of $K_4$ geometrically. For example, $c(K_6,K_4)=3 <
4=\rho^2_3(K_6)$ (see Theorem~\ref{thm:rhoCompleteGraphs}).

In order to determine $\rho^2_3(K_n)$ for particular values of $n$,
we need some properties of geometric and
combinatorial covers of~$K_n$.

\newcommand{\contentLemKnCovTriangle}{%
  Let $d \colon V(K_n)\to\mathbb R^3$ be a crossing-free drawing of $K_n$ and $\mathcal L$ a
  geometric cover of $d$.
  For each 4-vertex graph $G_\ell\in\mathcal K_{\mathcal L}$,
  the set $d(G_\ell)$ not only belongs to a plane~$\ell$, but also
  defines a triangle with an additional vertex in its interior.}
\wormhole{lem:KnCovTriangle}
\begin{lemma}
  \label{lem:KnCovTriangle}
  \contentLemKnCovTriangle
\end{lemma}

\newcommand{\contentLemKnCovCommon}{%
  Let $d \colon V(K_n)\to\mathbb R^3$ be a crossing-free drawing of $K_n$ and $\mathcal L$ a
  geometric cover of $d$.
  No two different 4-vertex graphs $G_\ell,G_{\ell'}\in\mathcal K_{\mathcal
  L}$ can have three common vertices.
}
\wormhole{lem:KnCovCommon}
\begin{lemma}
  \label{lem:KnCovCommon}
  \contentLemKnCovCommon
\end{lemma}

\newcommand{\contentThmRhoCompleteGraphs}{%
  For $n \le 9$, the value of $\rho^2_3(K_n)$ is bounded by the numbers in
  Table~\ref{tab:Kn}.
}
\wormhole{thm:rhoCompleteGraphs}
\begin{theorem}
  \label{thm:rhoCompleteGraphs}
  \contentThmRhoCompleteGraphs
  \begin{table}[tbh]
    \caption{Lower and upper bounds for $\rho^2_3(K_n)$ for small
      values of~$n$.}
    \label{tab:Kn}
    \centering
    \setlength{\tabcolsep}{0.175cm}
    \begin{tabular}{crrrrrr}\hline\noalign{\smallskip}
      $n$   & 4 & 5 & 6 & 7 & 8 & 9 \\ \midrule
      $\ge$ & 1 & 3 & 4 & 6 & 6 & 7 \\
      $\le$ & 1 & 3 & 4 & 6 & 7 &   \\ \bottomrule
    \end{tabular}
  \end{table}
\end{theorem}

\begin{proof}
  Here, we show only the bounds for $n=6$.  For the remaining proofs,
  see Appendix~\appref{sec:gen-appendix}.
  Fig.~\ref{fig:K6-photo} shows that $\rho^2_3(K_6)\le 4$. Now we
  show that $\rho^2_3(K_6)\ge 4$. Assume that $\rho^2_3(K_6)<4$.
  Consider a combinatorial cover $\mathcal
  K_{\mathcal L}$ of $K_6$ by its complete planar subgraphs
  corresponding to a geometric cover $\mathcal L$ of its drawing by
  $3$ planes.
  Graph $K_6$ has $15$ edges, so to cover it by complete
  planar graphs we have to use at least two copies of $K_4$ and,
  additionally, a copy of $K_k$ for $3 \leq k \leq 4$. But, since each two
  copies of~$K_4$ in~$K_6$ have a common edge (and by
  Lemma~\ref{lem:KnCovCommon} this edge is unique), the cover
  $\mathcal K_{\mathcal L}$ consists of three copies of $K_4$.
  Denote these copies by $K_4^1$, $K_4^2$, and $K_4^3$. By
  Lemma~\ref{lem:KnCovTriangle}, for each $i$, $d(K_4^i)$ is a triangle
  with an additional vertex $d(v_i)$ in its interior.  Let
  $V_0=\{v_1,v_2,v_3\}$.  By the Krein--Milman theorem~\cite{km-eprcs-SM40,Wik4},
  the convex hull $\Conv(d(K_6))$ is the convex hull
  $\Conv(d(V(K_6))\setminus d(V_0))$. If
  all the vertices $v_i$ are mutually distinct then the set
  $d(V(K_6))\setminus d(V_0)$ is a triangle, so the drawing $d$ is planar,
  a contradiction.  Hence, $v_i=v_j$ for some $i\ne j$.  Let $k$ be
  the third index that is distinct from both $i$ and $j$. Since
  graphs $K_4^i$ and $K_4^j$ have exactly one common edge, this is an
  edge $(v_i,v)$ for some vertex $v$ of $K_6$ (see
  Fig.~\ref{fig:K6-photo} with $u_4$ for $v_i$ and $u_1$ for $v$).
  Let $V(K_4^i)=\{v,v_i,v_i^1,v_i^2\}$ and
  $V(K_4^j)=\{v,v_j,v_j^1,v_j^2\}$.  Since the union $K_4^1\cup
  K_4^2\cup K_4^3$ covers all edges of~$K_6$, all edges
  $(v_i^1,v_j^1)$, $(v_i^1,v_j^2)$, $(v_i^2,v_j^1)$, and
  $(v_i^2,v_j^2)$ belong to~$K_4^k$. Thus
  $V(K_4^k)=\{v_i^1,v_i^2,v_j^1,v_j^2\}$.  But vertices $v_i^1$,
  $v_i^2$, $v_j^1$, and $v_j^2$ are in convex position (see
  Fig.~\ref{fig:K6-photo}), a contradiction to
  Lemma~\ref{lem:KnCovTriangle}.
\end{proof}

\begin{question}
Estimate $\rho^2_3(G)$ (and $\pi^2_3(G)$?) from above for toroidal and planar projective graphs.
\end{question}

\begin{question}
$\rho^2_3(G)=O(1)$ for cubic graphs?
\end{question}

\section{The Affine Cover Numbers of Planar Graphs
  (\texorpdfstring{$\reals^2$ and $\reals^3$}{R² and R³})}\label{s:planar}

\subsection{Placing Vertices on Few Lines
  (\texorpdfstring{$\pi^1_2$ and $\pi^1_3$}{π12 and π13})}%

\savespace{
Call a drawing \emph{outerplanar} if all the vertices lie on
the outer face.  An \emph{outerplanar graph} is a graph admitting
an outerplanar drawing.  Note that this definition does not depend on
whether straight line or curved drawings are considered.
}

Combining Corollary~\ref{PiChi} with the 4-color theorem yields
$\pi^1_3(G) \le 4$ for planar graphs.
Given that outerplanar graphs are 3-colorable (they are partial
2-trees), we obtain $\pi^1_3(G) \le 3$ for these graphs.
These bounds can be improved using
the equality $\pi^1_3(G)=\lva G$ of Theorem \ref{PiChar0}
and known results on the linear vertex arboricity:
  \begin{enumerate}[(a)]
  \item For any planar graph~$G$, it holds that $\pi^1_3(G) \le 3$
    \cite{Goddard91,Poh90}.
  \item \label{enum:chartrand-kronk}
    There is a planar graph $G$ with $\pi^1_3(G)=3$ \cite{ck-papg-JLMS69}.
  \item \label{enum:outerplanar}
    For any outerplanar graph $G$, $\pi^1_3(G) \le 2$
    \cite{AkiyamaEGW89,BroereM85,Wang88}.
  \end{enumerate}

According to Chen and He~\cite{ChenH96}, the
upper bound $\lva{G}\le 3$ for planar graphs by Poh~\cite{Poh90} is
constructive and yields a polynomial-time algorithm for partitioning
the vertex set of a given planar graph into three parts, each inducing
a linear forest.
By combining this with the construction given in Theorem~\ref{PiChar0}, we
obtain a polynomial-time algorithm that draws a given planar graph such
that the vertex set ``sits'' on three lines.

The example of Chartrand and Kronk~\cite{ck-papg-JLMS69}
is a 21-vertex planar graph whose \emph{vertex arboricity} is~3,
which means that the vertex set of this graph cannot even be split
into two parts both inducing (not necessarily linear) forests.
Raspaud and Wang \cite{RaspaudW08} showed that all 20-vertex planar graphs have vertex arboricity at most~2.
We now observe that a smaller example of a planar graph attaining the extremal value $\pi^1_3(G)=3$
can be found by examining the \emph{linear} vertex arboricity.

\newcommand{\contentRemPlanarLvaThree}{%
  The planar 9-vertex graph $G$
  in Fig.~\ref{fig:three-lines} has $\pi^1_3(G)=\lva G=3$.%
}
\wormhole{rem:PlanarLvaThree}
\begin{example}
  \label{rem:PlanarLvaThree}
  \contentRemPlanarLvaThree
\ (See a proof in Appendix~\appref{sec:plan-appendix}.)
\end{example}

\begin{question}\label{NineLVA}
  Is there a planar graph~$G$ with less than 9 vertices and
  $\lva{G}=3$?
\end{question}

Now we show lower bounds for the parameter~$\pi^1_2(G)$.
\savespace{
Recall that the \emph{dual} of a 3-connected planar graph $G$ is a graph $G^*$ whose
vertices are the faces of $G$ (represented by their facial cycles).
The definition does not depend on a particular embedding of $G$ in the plane
by the Whitney theorem, which says that all embeddings of a 3-connected planar graph
in the sphere are equivalent up to a homeomorphism
(therefore, the set of facial cycles of~$G$ does not
depend on a particular plane embedding).
Two faces are adjacent in $G^*$ iff they share a common edge.  The
dual graph~$G^*$ is also a polyhedral graph, and $(G^*)^*$, is
isomorphic to~$G$.  In a cubic graph, every vertex has
degree~3; the dual of any cubic 3-connected planar graph is a triangulation.
Conversely, the dual of a triangulation is a cubic graph.
}

Recall that the \emph{circumference} of a graph $G$, denoted by $c(G)$, is the
length of a longest cycle in $G$.
For a planar graph $G$, let $\bar v(G)$ denote the maximum $k$ such that
$G$ has a straight-line plane drawing with $k$ collinear vertices.

\begin{lemma}\label{lem:CycleCover} Let $G$ be a planar graph. Then $\pi^1_2(G)\ge n/\bar v(G)$. If $G$ is a triangulation then
    $\pi^1_2(G)\ge (2n-4)/c(G^*)$.
\end{lemma}
\begin{proof} Since the first claim is obvious, we prove only the second.
  Let $\gamma(G)$ denote the minimum number of cycles in the dual graph~$G^*$
  sharing a common vertex and covering every vertex of $G^*$ at least twice.
  Note that, as $G$ is a triangulation, $\gamma(G)\ge(4n-8)/c(G^*)$,
  where $2n-4$ is the number of vertices in $G^*$ (as a consequence of Euler's
  formula).
  We now show $\pi^1_2(G)\ge\gamma(G)/2$, which implies the claimed result.

  Given a drawing realizing $\pi^1_2(G)$ with line set~$\mathcal L$,
  for every line $\ell \in \mathcal L$, draw two parallel
  lines~$\ell',\ell''$ sufficiently close to $\ell$ such that they
  together intersect the interiors of all faces touched by~$\ell$
and do not go through any vertex of the drawing.
Note that $\ell'$ and $\ell''$ cross boundaries of faces only via inner points of edges.
Each such crossing corresponds to a transition from one vertex to another
along an edge in the dual graph $G^*$. Since all the faces of $G$ are triangles,
each of them is visited by each of $\ell'$ and $\ell''$ at most once.
Therefore, the faces crossed along $\ell'$ and the faces crossed along $\ell''$,
among them the outer face of~$G$, each form a cycle in~$G^*$.
It remains to note that every face $f$
of the graph $G$ is crossed at least twice, because $f$ is intersected by at least two different lines from
$\mathcal L$ and each of these two lines has a parallel copy that crosses~$f$.
\end{proof}

An infinite family of triangulations $G$
with $\bar v(G)\le n^{0.99}$ is constructed in~\cite{RavskyV11}.
By the first part Lemma \ref{lem:CycleCover} this implies that
there are infinitely many triangulations $G$ with $\pi^1_2(G)\ge n^{0.01}$.
The second part of Lemma \ref{lem:CycleCover} along with an estimate of
Gr\"unbaum and Walther \cite{GruenbaumW73} (that was used also in~\cite{RavskyV11})
yields a stronger result.

\begin{theorem}\label{thm:pi12large}
There are infinitely many triangulations $G$ with $\Delta(G)\le12$
and $\pi^1_2(G)\ge n^{0.01}$.
\end{theorem}

\begin{proof}
The \emph{shortness exponent} $\sigma_{\mathcal G}$
of a class $\mathcal G$ of graphs is the infimum of the set of the reals
$\liminf_{i\to\infty}{\log c(H_i)}/{\log |V(H_i)|}$
for all sequences of $H_i\in\mathcal G$ such that $|V(H_i)|<|V(H_{i+1})|$.
Thus, for each~$\epsilon>0$, there are infinitely many graphs $H\in\mathcal G$
with $c(H)<|V(H)|^{\sigma_{\mathcal G}+\epsilon}$.
The dual graphs of triangulations with maximum vertex degree at most 12
are exactly the cubic 3-connected planar graphs with each face incident to
at most 12 edges (this parameter is well defined by the Whitney theorem).
Let $\sigma$ denote the shortness exponent for this class of graphs.
It is known \cite{GruenbaumW73} that $\sigma\le\frac{\log26}{\log27}=0.988\ldots$.
The theorem follows from this bound by the second part of~Lemma \ref{lem:CycleCover}.
\end{proof}

\begin{problem}
Does $\pi^1_2(G)=o(n)$ hold for all planar graphs $G$?
\end{problem}

\begin{question}
  Let~$G$ be a planar graph with $\maxdeg(G)=O(1)$.
Does this imply $\pi^1_2(G)=o(n)$?
Does $\maxdeg(G)=3$ imply $\pi^1_2(G)=O(1)$?
Note that the proof of Theorem \ref{thm:pi12large} cannot be extended to
planar graphs of maximum vertex degree 6 because the shortness exponent
of the cubic 3-connected planar graphs with each faces incident to
at most 6 edges is known to be equal to~1~\cite{Ewald73}.
\end{question}

A \emph{track drawing}~\cite{flw-sldri-JGAA03} of a graph is
a plane drawing for which there are parallel lines, called \emph{tracks}, such that
every edge either lies on a track or its endpoints lie on two consecutive tracks.
We call a graph \emph{track drawable} \label{page:track} if it has a track drawing.
Let $\tn(G)$ be the minimum number of tracks of a track drawing of~$G$. Note that $\pi^1_2(G)\le\bar\pi^1_2(G)\le\tn(G)$.

The following proposition is similar to a lemma of
Bannister et al.~\cite[Lemma~1]{bddew-tllpd-arXiv15} who say it is
implicit in the earlier work of Felsner et al.~\cite{flw-sldri-JGAA03}.

\begin{theorem}[cf.~\cite{flw-sldri-JGAA03,bddew-tllpd-arXiv15}]
  \label{PiTrDr}
  Let $G$ be a track drawable graph. Then $\pi^1_2(G)\le 2$.
\end{theorem}
\begin{wrapfigure}[6]{r}{0.2\linewidth}
  \includegraphics[page=3,scale=1.0,bb=0 0 23mm 16mm]{trdr}
\end{wrapfigure}

\noindent\emph{Proof.}\nobreakspace
  Consider a track drawing of~$G$, which we now transform to a drawing on two
  intersecting lines.
  Put the tracks consecutively along a spiral so that they correspond to
  disjoint intervals on the half-lines as depicted on the right.
  Tracks whose indices are equal modulo~4 are placed on the same half-line;
  for more details see Fig.~\appref{fig:trdr} in Appendix~\appref{sec:plan-appendix} on
  page~\apppageref{fig:trdr}.
  (Bannister et al.~\cite[Fig.~1]{bddew-tllpd-arXiv15} use three half-lines
  meeting in a point.)
\qed

Observe that any tree is track drawable:
two vertices are aligned on the same track
iff they are at the same
distance from an arbitrarily assigned root.
Moreover, any outerplanar graph is track drawable~\cite{flw-sldri-JGAA03}.
This yields an improvement over
the bound $\pi^1_3(G) \le 2$ for outerplanar graphs stated in the beginning of this section.

\begin{corollary}\label{Pi12Outerplanar}
  For any outerplanar graph~$G$, it holds that $\pi^1_2(G)\le 2$.
\end{corollary}

\savespace{
Next, we consider a generalization of trees, the class of \emph{2-trees}, which is recursively defined as follows:
\begin{itemize}
\item
the graph consisting of two adjacent vertices is a 2-tree;
\item
if $G$ is a 2-tree and $H$ is obtained from $G$ by adding a new
vertex and connecting it to two adjacent vertices of $G$,
then $H$ is a 2-tree.
\end{itemize}
A graph is a \emph{partial 2-tree} if it is a subgraph of a 2-tree.
It is well known that the class of partial 2-trees coincides
with the class of graphs with treewidth at most 2.
Any outerplanar graph is a partial 2-tree, and the same holds
for series-parallel graphs (the latter class is sometimes defined
so that it coincides with the class of partial 2-trees).
Note that not all 2-trees are track drawable (for example, the graph
consisting of three triangles that share one edge).

Ravsky and Verbitsky~\cite[Theorem 4.5]{RavskyV11} have shown that any
partial 2-tree admits a drawing with at least $n/30$ collinear
vertices. This suggests the following question.

\begin{problem}%
  $\pi^1_2(G)=O(1)$ for 2-trees?
\end{problem}
}

\subsection{Placing Edges on Few Lines
  (\texorpdfstring{$\rho^1_2$ and $\rho^1_3$}{ρ12 and ρ13})}\label{ss:rho12rho13}

The parameter $\rho^1_2(G)$ is related to two parameters introduced by Dujmovi{\'c}
et al.~\cite{desw-dpgfs-CGTA07}.
They define a \emph{segment} in a straight-line drawing of a graph $G$
as an inclusion-maximal (connected) path of edges of $G$ lying on a line.
A \emph{slope} is an inclusion-maximal set of parallel segments.
The \emph{segment number} (resp., \emph{slope number}) of a planar graph $G$ is the minimum possible
number of segments (resp., slopes) in a straight-line drawing of $G$.
We denote these parameters by $\segm(G)$ (resp., $\slop(G)$).
Note that $\slop(G)\le\rho^1_2(G)\le\segm(G)$.

These parameters can be far away from each other.
Figure~\ref{fig:nested-triangles} shows a graph with $\slop(G)=O(1)$ and $\rho^1_2(G)=\Omega(n)$
(see the proof of Theorem~\ref{thm:nested-triangles}).
On the other hand, note that $\rho^1_2(mK_2)=1$ while $\segm(mK_2)=m$
where $mK_2$ denotes the graph consisting of $m$ isolated edges.
The gap between $\rho^1_2(G)$ and $\segm(G)$ can be large even for
connected graphs. It is not hard to see that~$\segm(G)$ is bounded from below
by half the number of odd degree vertices (see~\cite{desw-dpgfs-CGTA07} for details). Therefore, if we take a caterpillar $G$
with $k$ vertices of degree 3 and $k+2$ leaves, then $\segm(G)\ge n/2$, while
$\rho^1_2(G)=O(\sqrt n)$ because $G$ can easily be drawn in a square grid of area $O(n)$.
Note that, for the same $G$, the gap between $\slop(G)$ and $\rho^1_2(G)$ is also
large. Indeed, $\slop(G)=2$ while $\rho^1_2(G)>\sqrt{n-2}$ by Lemma~\ref{Ess}(a).

It turns out that a large gap between $\rho^1_2(G)$ and $\segm(G)$ can be shown
also for 3-connected planar graphs and even for triangulations.

\begin{example}\label{ex:rho12-segm}
  \textit{There are triangulations with $\rho^1_2(G)=O(\sqrt n)$ and $\segm(G)=\Omega(n)$.}\footnote{%
A triangulation~$G$ with $\segm(G)=O(\sqrt{n})$
has been found by Dujmovi{\'c} at al.~\cite[Fig.~12]{desw-dpgfs-CGTA07}.}
Note that this gap is the best possible because any 3-connected graph~$G$
has minimum vertex
degree~$3$ and, hence, $\rho^1_2(G)\ge\rho^1_3(G)>\sqrt{2n}$ by Lemma~\ref{Ess}(a).
Consider the graph shown in Fig.~\ref{fig:Sashas-pattern}.
Its vertices are placed on the standard orthogonal grid and two slanted grids,
which implies that at most $O(\sqrt n)$ lines are involved.
The pattern can be completed to a triangulation by adding three vertices around it
and connecting them to the vertices on the pattern boundary.
Since the pattern boundary contains $O(\sqrt n)$ vertices, $O(\sqrt n)$
new lines suffice for this. Thus, we have $\rho^1_2(G)=O(\sqrt n)$ for
the resulting triangulation $G$. Note that the vertices drawn fat in Fig.~\ref{fig:Sashas-pattern}
have degree~$5$, and there are linearly many of them.
This implies that $\segm(G)=\Omega(n)$.
\end{example}

\begin{figure}[tb]
\begin{minipage}[b]{.35\textwidth}
  \centering
  \includegraphics[scale=0.7]{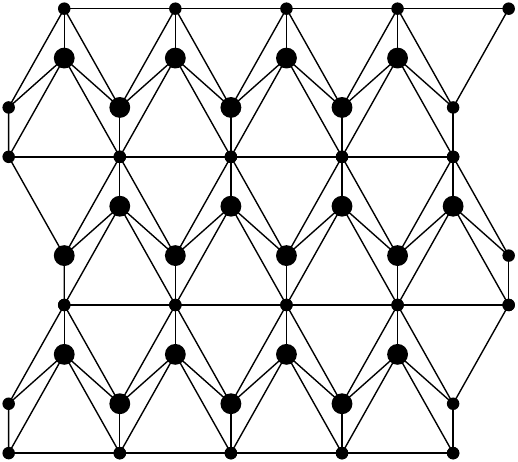}
  \caption{The main body of a triangulation~$G$ with
    $\rho^1_2(G)=O(\sqrt n)$ and $\segm(G)=\Omega(n)$.}
  \label{fig:Sashas-pattern}
\end{minipage}
\hfill
\begin{minipage}[b]{.25\textwidth}
    \centering
    \begin{tikzpicture}[scale=0.48]
      \foreach \deg in {90,210,330} {
        \foreach \dist in {0.5,1.5,2.5,3.5} {
          \fill (\deg:\dist) circle(4pt);
        }
        \draw (\deg:0.5) -- (\deg:3.5);
      }
      \foreach \dist in {0.5,1.5,2.5,3.5} {
        \draw (90:\dist) -- (210:\dist) -- (330:\dist) -- cycle;
      }
    \end{tikzpicture}
    \caption{The nested-triangles graph $T_k$.}%
    \label{fig:nested-triangles}
\end{minipage}
\hfill
\begin{minipage}[b]{.35\textwidth}
    \centering
    \includegraphics[page=2]{C4-Pk-in-3d-grid}
    \caption{Sketch of the construction in the proof of
      Theorem~\ref{thm:nested-triangles}(b).}
    \label{fig:3D-grid-small}
\end{minipage}
\end{figure}

Somewhat surprisingly, the parameter $\segm(G)$ can be bounded from above
by a function of $\rho^1_2(G)$ for all connected graphs.

\newcommand{\contentThmSegmRho}{%
  For any connected planar graph~$G$, $\segm(G)=O(\rho^1_2(G)^2)$.%
}
\wormhole{thm:segm-vs-rho12}
\begin{theorem}\label{thm:segm-vs-rho12}
\contentThmSegmRho
\end{theorem}

Note that
$
\maxdeg(G)/2 \le \rho^1_3(G) \le \rho^1_2(G) \le \segm(G) \le m
$
for any planar graph~$G$.
For all inequalities here except the second one, we already know
that the gap between the respective pair of parameters can be very large
(by considering a caterpillar with linearly many degree 3 vertices and
applying Lemma~\ref{Ess}(a), by Example~\ref{ex:rho12-segm}, and
by considering the path graph $P_n$, for which $\segm(P_n)=1$).
Part (b) of the following theorem shows a large gap also between the
parameters $\rho^1_3(G)$ and~$\rho^1_2(G)$, that is, some planar graphs can be
drawn much more efficiently, with respect to the line
cover number, in 3-space than in the plane.

\begin{theorem}\label{thm:nested-triangles}
\begin{enumerate}[(a)]
\item There are infinitely many planar graphs with constant
  maximum degree, constant treewidth, and linear $\rho^1_2$-value.%
\item For infinitely many $n$ there is a planar
  graph $G$ on $n$ vertices with $\rho^1_2(G)=\Omega(n)$ and $\rho^1_3(G)=O(n^{2/3})$.
\end{enumerate}
\end{theorem}

\begin{proof}
Consider the nested-triangles graph $T_k=C_3\times P_k$ shown in Fig.~\ref{fig:nested-triangles}.
To prove statements~(a) and~(b), it suffices to establish the following bounds:
\begin{enumerate}[(i)]
\item \label{nested-triang:bound-lower}
$\rho^1_2(T_k) \ge n/2$ and
\item \label{nested-triang:bound-upper}
$\rho^1_3(T_k)=O(n^{2/3})$.
\end{enumerate}
To see the linear lower bound~\eqref{nested-triang:bound-lower}, note that~$T_k$ is 3-connected.
Hence, Whitney's theorem implies that, in any plane drawing of~$T_k$,
there is a sequence of nested triangles of length at least $k/2$.
The sides of the triangles in this sequence
must belong to pairwise different lines.
Therefore, $\rho^1_2(T_k) \ge 3k/2=n/2$.

For the sublinear upper bound~\eqref{nested-triang:bound-upper},
first consider the graph $C_4\times P_k$.
We build wireframe rectangular prisms that are stacks of $O(\sqrt[3]{n})$ squares each.
These prisms are placed onto the base plane in an $O(\sqrt[3]{n}) \times
O(\sqrt[3]{n})$ grid; see Fig.~\ref{fig:3D-grid-small}.
So far we can place the edges on the $O(n^{2/3})$ lines of the 3D cubic grid of
volume $O(n)$.
Next, we construct a path that traverses all squares by passing through the
prisms from top to bottom (resp., vice versa) and connecting neighboring prims.
We rotate and move some of the squares at the top (resp., bottom) of the prisms
to be able to draw the edges between neighboring prisms according to this path.
For this ``bending'' we need $O(n^{2/3})$ additional lines.
In Appendix~\appref{sec:plan-appendix} we provide a
drawing; see Fig.~\appref{fig:3D-grid} on page~\apppageref{fig:3D-grid}.
The same approach works for the graph $T_k = C_3\times P_k$.
In addition to the standard 3D grid, here we need also its slanted, diagonal
version (and, again, additional lines for bending in
the cubic box of volume $O(n)$).
The number of lines increases just by a constant factor.
\end{proof}

\savespace{
\begin{problem}%
Recall that, by Theorem~\ref{thm:rho^1_3-small-edge-separators},
graphs of bounded degree can have linearly large parameter $\rho^1_3(G)$.
  If a \emph{planar} graph $G$ has bounded degree, is it true that $\rho^1_3(G)=o(n)$?
\end{problem}
}

We are able to determine the exact values of $\rho^1_2(G)$ for complete bipartite graphs $K_{p,q}$
that are planar.

\newcommand{\contentExKTwoNMinusTwo}{%
\textit{$\rho^1_2(K_{1,q})=\lceil m/2\rceil$ and $\rho^1_2(K_{2,q})=\lceil (3n-7)/2\rceil=\lceil (3m-2)/4\rceil$.}
}
\wormhole{K2n-2}
\begin{example}\label{K2n-2}
\contentExKTwoNMinusTwo
See Appendix~\appref{sec:plan-appendix} for details.
\end{example}

\begin{question}
$\rho^1_3(K_{2,q})<\rho^1_2(K_{2,q})$?
\end{question}

Motivated by Example~\ref{K2n-2}, we ask:
\begin{problem}%
  What is the smallest $c$ such that $\rho^1_2(G)\le (c+o(1)) m$
  for any planar graph~$G$? Example~\ref{K2n-2} shows that $c\ge 3/4$.
  Durocher and Mondal~\cite{dm-dptfs-CCCG14}, improving on an earlier
  bound of Dujmovi\'c et al.~\cite{desw-dpgfs-CGTA07}, showed that
  $\segm(G)<\frac73n$ for any planar graph~$G$. This implies that $c\le 7/9$.
\end{problem}%

For any binary tree~$T$, it holds that $\rho^1_2(T)=O(\sqrt{n\log n})$.
This follows from the known fact~\cite{cgkt-oaars-CGTA02} that $T$
has an orthogonal drawing on a grid of size $O(\sqrt{n\log n})
\times O(\sqrt{n\log n})$.
For complete binary trees lower and upper bounds are described in
Example~\appref{prop:CompleteBinTreeRho12} in Appendix~\appref{sec:edges-appendix}.

\begin{question}
  $\rho^1_3(G)=\rho^1_2(G)$  trees?
\end{question}

\begin{question}
  Is it true that for trees $\Delta(G)=O(1)$ implies $\rho^1_2(G)=o(1)$?
More specifically, $\rho^1_2(G)=O(n^{1-1/\Delta(G)})$?
\end{question}

\bibliographystyle{splncs03}
\bibliography{abbrv,planes}

\clearpage
\appendix

\noindent{\Large\bfseries Appendix}

\section{Collapse of the Multidimensional Affine Hierarchy}
\label{s:collapse}

\begin{theorem}\label{dCollapse}
  For any integers $1 \le l \le d$, $d\ge 3$, and for any graph $G$,
  it holds that $\pi^l_d(G)=\pi^l_3(G)$,
  $\bar\pi^l_d(G)=\bar\pi^l_3(G)$, and $\rho^l_d(G)=\rho^l_3(G)$.
\end{theorem}

\begin{proof}
  The theorem follows from the following fact: For any finite family
  $\mathcal L$ of lines in~$\reals^d$, there exists a linear
  transformation $A \colon \reals^d\to\reals^3$ that is injective
on $L_0=\bigcup\setdef{\ell}{\ell\in\mathcal L}$,
the set of all points of the lines in~$\mathcal L$.

  We prove this claim by induction on~$d$.  If $d=3$, we let $A$ be
  the identity map on $\reals^3$. Suppose that $d>3$.

Regarding two lines $\ell$ and $\ell'$ in $\reals^d$ as 1-dimensional
affine subspaces, we consider the Minkowski difference $\ell - \ell'=\setdef{l-l'}{l\in\ell,\,l'\in\ell'}$.
Note that this is a plane, i.e., a 2-dimensional affine subspace of~$\reals^d$.
Denote $L=L_0-L_0=\bigcup\setdef{\ell-\ell'}{\ell,\ell'\in \mathcal L}$.
Let $L'=\setdef{tl}{t\in\reals,\,l\in L}$ be the union of all lines going through the
origin $0$ of $\mathbb R^d$ and intersecting the set $L$.
  Since the set $L$ is contained in a union of finitely many planes in the
  space $\mathbb R^d$, the set $L'$ is contained in a union of finitely
  many $3$-dimensional linear subspaces of $\reals^d$ (each of them contains the origin $0$).
Since $d>3$, there exists a line $\ell_0\ni 0$ such
  that $\ell_0\cap L=\{0\}$.  Now, let $A_0\function{\reals^d}{\reals^{d-1}}$
be an arbitrary linear transformation such that $\ker A_0=\ell_0$.
  Let $x,y\in L_0$ be arbitrary points such that $A_0x=A_0y$.
  Then $x-y\in (L_0-L_0)\cap\ker A_0=L\cap\ell_0=\{0\}$, so $x=y$.
  Thus, the map $A_0$ is injective on~$L_0$.

  By the inductive assumption applied to the family of lines
  $\setdef{A_0\ell}{\ell\in\mathcal L}$ in $\reals^{d-1}$, there exists a linear
  transformation $A_1 \function{\reals^{d-1}}{\reals^3}$
injective on the union of all $\ell\in\mathcal L$.
It remains to take the composition $A=A_1A_0$.
\end{proof}

\section{The Parallel Affine Cover Numbers}
\label{s:parallel}

\subsection{Placing Vertices on Few Parallel Lines in 3-Space
  (\texorpdfstring{$\bar\pi^1_3$}{‾π13})}
\label{barPi^1_3}

The concept of a proper track drawing was introduced by Dujmovi\'c et
al.~\cite{dpw-tlg-DMTCS04} in combinatorial terms with the following
geometric meaning.  We call a 3D drawing of a graph~$G$ a
\emph{proper track drawing} if there are parallel lines, called
\emph{tracks}, such that every vertex of $G$ lies on one of the tracks
and every edge connects vertices lying on two different tracks. Edges
between any two tracks are not allowed to cross each other.
Furthermore, we call a drawing of~$G$ an \emph{improper track drawing} if
we allow edges between consecutive vertices of a same track.  The
\emph{proper track number}~$\tn^{\mathrm{p}}(G)$ (\emph{improper track
  number}~$\tn^{\mathrm{i}}(G)$) of~$G$ is the minimum number of
tracks of a proper (improper) track drawing of~$G$.
While the definition above does not excludes crossing of two edges if they are
between disjoint pairs of tracks, note that all such crossings can be removed
by slightly shifting the vertices within each track. We, therefore, have
$$
\bar\pi^1_3(G)=\tn^{\mathrm{i}}(G)
$$
for any graph~$G$.
By~\cite[Lemma 2.2]{dmw-lgbtw-SICOMP05},
$\tn^{\mathrm{p}}(G)/2\le\tn^{\mathrm{i}}(G)\le\tn^{\mathrm{p}}(G)$.
Therefore, the upper bounds for $\tn^{\mathrm{p}}(G)$ surveyed by Dujmovi{\'c} et
al.~\cite[Table~1]{dpw-tlg-DMTCS04}) imply also upper
bounds on $\bar\pi^1_3(G)$ for different classes of graphs~$G$.

In particular, Dujmovi{\'c}, Morin, and Wood \cite{dmw-lgbtw-SICOMP05} prove that
$$
\tn^{\mathrm{p}}(G)\le 3^{\tw(G)}\cdot 6^{(4^{\tw(G)}-3{\tw(G)}-1)/9}
$$
for any graph $G$.
Note that any upper bound for $\bar\pi^1_3(G)$ implies also an upper bound for~$\rho^2_3(G)$.

\begin{lemma}\label{RhoByPi}
  $\rho^2_3(G)\le \binom{\bar\pi^1_3(G)}{2}$.
\end{lemma}

\begin{proof}
Any two parallel lines of the drawing lie in a plane, and all the edges are located in these planes.
\end{proof}

Since $\bar\pi^1_3(G)\le\tn^{\mathrm{p}}(G)$, Lemma \ref{RhoByPi}
implies that the parameter $\rho^2_3(G)$ is bounded from above by a function of
the treewidth of~$G$.

Whether or not $\tn^{\mathrm{p}}(G)$ and, hence, $\bar\pi^1_3(G)$
is bounded for the class of planar graphs is a
long-standing open problem; see~\cite{dmw-lgbtw-SICOMP05}.  The best lower bound is $\bar\pi^1_3(G)\ge 7$; see
\cite[Cor.~10]{dpw-tlg-DMTCS04}.
The best upper bound $\bar\pi^1_3(G)=O(\sqrt n)$ for planar graphs is obtained by Dujmovi\'c and Wood~\cite{DujmovicW04}.
Wood~\cite[Theorem~2]{w-bdgalqn-DMTCS08} proved that for all $\Delta\ge 3$ and for all
sufficiently large $n > n(\Delta)$, there is a simple $\Delta$-regular
$n$-vertex graph with $\tn^{\mathrm{p}}(G)\ge c\sqrt{\Delta}n^{1/2-1/\Delta}$ for
some absolute constant $c$, which implies that $\bar\pi^1_3(G)$
is unbounded even over graphs of bounded degree.

\begin{question}
These results suggest the following relaxation
of the open problem about $\bar\pi^1_3(G)$ for planar graphs.
  Let $G$ be a {\it planar} graph and $\maxdeg(G)=O(1)$.
  Is $\bar\pi^1_3(G)=O(1)$?
\end{question}

\begin{theorem}\label{PlanarBarPi}
  \begin{enumerate}
  \item[(a)]
If $\bar\pi^1_3(G)\le 3$ then $G$ is planar and $\pi^1_2(G)\le \bar\pi^1_3(G)$.
\item[(b)]
If $\pi^1_2(G)\le2$ for a planar graph $G$, then $\bar\pi^1_3(G)\le 4$.
\item[(c)]
If $G$ is, moreover, track drawable,\footnote{%
in the sense of the definition in page~\pageref{page:track}.}
then $\bar\pi^1_3(G)\le 3$.
  \end{enumerate}

\end{theorem}
\begin{proof}
(a)
  Since any two parallel lines in~$\mathbb R^3$
  define a plane, the graph $G$ can be drawn on
  the three rectangular faces of a triangular prism in $\mathbb{R}^3$,
  created by (a projection of) the parallel lines. Therefore, $G$ is
  planar.  Now pick a point outside of the prism, but close to its
  triangular base.  Project the drawing of~$G$ on a plane that does
  not intersect the prism, is close and parallel to the other
  triangular face of the prism.  This yields a drawing of~$G$
  whose vertices lie on $\bar\pi^1_3(G)$ lines.

(b)
Assume that the graph $G$ is drawn on two lines $\ell_1$ and $\ell_2$.
If these lines are parallel then $\bar\pi^1_3(G)\le \bar\pi^1_2(G)\le 2$ and we are done.
If the lines $\ell_1$ and $\ell_2$ intersect in a point $O$,
then the union $\ell_1\cup\ell_2$ is split into three open rays and one half-open ray.
We can put the vertices from the rays into four parallel tracks,
preserving their order from the point $O$ to infinity along the rays.

(c)
This part is a version of Theorem~\ref{PiTrDr}.
\end{proof}

\begin{question}
Can Part (b) of Theorem \ref{PlanarBarPi} be improved to $\bar\pi^1_3(G)\le 3$?
\end{question}

\begin{example}
  The parallel affine cover number of complete (bipartite)
  graphs is as follows.
  \begin{enumerate}[(a)]
  \item \textit{$\bar\pi^1_3(K_n) = n-1$ for any $n \ge 2$.}
  \item \textit{$\bar\pi^1_3(K_{p,q})=p+1$ for any $p\le q$ and $q\ge 3$.}
  \end{enumerate}
Part (a) is straightforward.

Let us prove Part (b) (for the proper track number the case $p=q$ was
considered by Dujmovi\'c and Whitesides \cite{DujmovicW13})
To show the upper bound $\bar\pi^1_3(K_{p,q})\le p+1$,
we put the independent set of $q$ vertices in a line
and use a separate line for each of the other $p$ vertices.
To show the lower bound,
let $\mathcal{L}$ be an optimal set of lines. Suppose that our bipartition is defined by $p$ white and $q$ black vertices.
First, suppose that one line $\ell \in \mathcal{L}$ contains a pair of monochromatic vertices.
Since our lines are parallel, no other line may contain two vertices of the other color.
Clearly, if $\ell$ is monochromatic, there are at least $p+1$ lines. If $\ell$ is not monochromatic,
then $\ell$ contains exactly three vertices, where the monochromatic pair is separated by the other vertex.
However, now, every line $\ell' \in \mathcal{L} \setminus \{\ell\}$ must be monochromatic,
otherwise the edges spanning between $\ell$ and $\ell'$ will produce a crossing.
Thus since $q \geq 3$ the total number of lines is at least $1+p-1+1\ge p+1$.

In the other case, no line contains two monochromatic vertices. If $p < q$, then we already have
our lower bound of $p+1$. However, if $p=q$, we need to argue a little more carefully.
Here, we note that there cannot be three lines with two vertices each since this would imply a crossing.
Note that, in order to avoid a crossing between the pair of edges connecting two lines, the order of
the colors on the first line must be reversed on the second. In particular, with three lines, some
 pair of lines will violate this condition. Thus the total number of lines is at least $2+q-2+p-2\ge p+1$, because $q \geq 3$.

\end{example}

\subsection{Placing Vertices on Few Parallel Lines in the Plane
  (\texorpdfstring{$\bar\pi^1_2$}{‾π12})}\label{barPi^1_2}

All graphs considered in this subsection are supposed to be planar.

Let $G$ be a planar graph with $\maxdeg(G)\le 2$.  Observe that $G$ is a union of
cycles and paths and, hence, $\pi^1_2(G)=\bar\pi^1_2(G)\le 2$.
However, if we relax the
degree restriction even just slightly to $\maxdeg(G)\le 3$,
the parameters $\pi^1_2(G)$ and $\bar\pi^1_2(G)$ can be different.
As a simplest examples, note that $\pi^1_2(K_4)=2$ while $\bar\pi^1_2(K_4)=3$.
In general, the gap is unbounded. In the following, we examine the gap for some interesting graph classes.

For any tree $G$,  we have $\pi^1_2(G)\le 2$ by Theorem~\ref{PiTrDr}.
On the other hand,
Felsner et al.~\cite{flw-sldri-JGAA03}
showed that $\bar\pi^1_2(G)\ge\log_3(2n+1)$ for every complete ternary tree~$G$.
We can show a much larger gap for graphs of vertex degree at most 3 with cycles.
Beforehand, we need some preliminaries.

Let $H$ be a \emph{plane} graph, that is, a planar graph drawn without edge crossings
in the plane. Removal of $H$ splits the plane into connected components,
which are called \emph{faces} of $H$. We define $H^f$ to be the graph whose vertices
are the faces of $H$; two faces are adjacent in $H^f$ if their boundaries have
a common point.
Note that this is not the same as the dual of~$G$,
which has an edge for each pair of faces that share an edge.
For a planar graph~$G$, let $\phi(G)=\min_H\diam(H^f)$, where
the minimum is taken over all plane representations~$H$ of~$G$.

\begin{lemma}\label{FaceExit}
$\bar\pi^1_2(G)\ge \phi(G)$.
\end{lemma}
\begin{proof}
Let $H$ be a drawing of $G$ attaining the value $r=\bar\pi^1_2(G)$.
The $r$ underlying lines partition the plane into $r+1$ parallel strips
(including two half-planes). If a face of $H$ intersects one of the bounded strips,
then it is incident to a vertex lying in a line above this strip.
This vertex is incident to another face intersecting a strip above this line.
The same holds true also in the downward direction.
It follows that
from each face we can reach the outer face in $H^f$ along a path of length at most $\lfloor r/2\rfloor$.
Therefore, $\diam(H^f)\le 2\lfloor r/2\rfloor\le r$.
\end{proof}

\begin{theorem}\label{SashaEx}
  For each $k$, there is a planar graph~$G$ on $n=4k$ vertices with $\maxdeg(G)\le 3$,
  $\pi^1_2(G)=2$, and $\bar\pi^1_2(G)\ge n/4$.
\end{theorem}

\begin{proof}
Consider the graph~$G=S_k$ consisting of $k=n/4$ nested copies of~$C_4$ connected as depicted in
the drawing~$H_k$ in Fig.~\ref{fig:nested-squares}.
This drawing certifies that~$\pi^1_2(S_k)=2$. Note that $\diam(H_k^f)=k$.
In order to apply Lemma~\ref{FaceExit}, we have to show that
this equality holds true as well for any other drawing of~$S_k$.

Note that $S_k$ is 2-connected. We use general facts about plane embeddings
of 2-connected graphs; here we do not restrict ourselves to straight-line
drawings only. A 2-connected planar graph $G$ can have many plane representations,
but all of them are obtainable from each other by a sequence of simple
transformations. Specifically, let $H$ be a plane version of $G$
and $C$ be a cycle in $H$ containing only two vertices, $u$ and $v$,
that are incident to some edges outside $C$. We can obtain another plane
embedding $H'$ of $G$ by \emph{flipping $G$ inside $C$}, that is,
by replacing the interior of $C$ with its mirror version (up to homeomorphism).
The rest of $G$ is unchanged; in particular, $u$ and $v$ keep their location.
It turns out \cite[Theorem 2.6.8]{MoharT-book} that, for any other plane
representation $H_1$ of $G$, $H$ can be transformed into $H_1$ by a sequence
of flippings that is followed, if needed, by re-assigning the outer face and
applying a plane homeomorphism.

Let us apply this to $G=S_k$ and its plane representation $H=H_k$.
First we need to identify all cycles in $H_k$ for which the flipping
operation is possible. Recall that such a cycle $C$ is connected
to its exterior only at two vertices $u$ and $v$. Clearly, removal
of these vertices disconnects the graph. This is possible only if
$u$ and $v$ belong to two diagonal edges forming a centrally symmetric pair
of edges. If the last condition is true for $u$ and $v$, then an appropriate
cycle $C$ exists only if the pair $\{u,v\}$ is centrally symmetric.
It readily follows from here that flipping is possible only with respect
to square cycles excepting the outer one.

Note now that, for each such cycle $C$, the interior of $C$ is
symmetric with respect to the axis passing through the corresponding
vertices $u$ and $v$. This implies that flipping of $H_k$ with respect
to $C$ does not change the graph $H^f_k$. Moreover, the flipped graph
differs from $H$ only by relabeling of vertices. Therefore,
any further flipping also does not change $H^f_k$.
Moreover, $H^f_k$ stays the same up to isomorphism after
re-assigning the outer face and applying a plane homeomorphism.
We conclude that $\phi(S_k)=\diam(H_k^f)=k$.
\end{proof}

\begin{figure}[tb]
    \centering
    \includegraphics{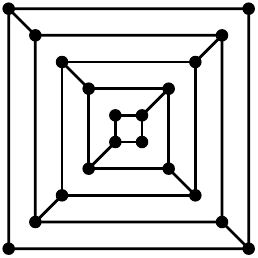}
    \caption{The nested-squares graph~$S_k$.}
    \label{fig:nested-squares}
\end{figure}

\begin{lemma}
  For any graph~$G$, it holds that $\bar\pi^1_2(G)\le \sqrt{\area(G)}$, where $\area(G)$ is the minimum
number of vertices in a rectangular grid containing a straight-line drawing of~$G$.
\end{lemma}

Di Battista and Frati~\cite{df-sadog-Algorithmica07} have shown that
outerplanar graphs can be drawn straight-line in area $O(n^{1.48})$,
which yields the following corollary.

\begin{corollary}[cf.~\cite{df-sadog-Algorithmica07}]
  For any outerplanar graph~$G$ with $n$ vertices, it holds that
  $\bar\pi^1_2(G)=O(n^{0.74})$.
\end{corollary}

\begin{question}%
  Is it true that, for any 2-tree~$G$, $\bar\pi^1_2(G)=o(n)$?
  What is the lower bound of $\bar\pi^1_2(G)$ for a 2-tree~$G$ in the
  worst case?
\end{question}

\section{The Affine Cover Numbers of General Graphs: Proofs}
\label{sec:gen-appendix}

A complete $r$-partite graph is called
\emph{balanced} if any two of its classes
differ by at most one in size.  Let $K^r(n)$ denote a
balanced complete $r$-partite graph with $n\ge r$ vertices.  In other
words, the vertex set of~$K^r(n)$ is split into $r$ disjoint classes,
$V_1,\dots,V_r$ such that $|V_i| \in \{\lfloor n/r\rfloor,\lceil
n/r\rceil\}$.

\begin{lemma}[{\cite[Lemma]{ptt-3dgdg-GD97}}]\label{pttDrawing}
  For any $r\ge2$ and for any $n$ divisible by $r$, the balanced
  complete $r$-partite graph $K^r(n)$ with vertex
  set~$V=V_1\cup\dots\cup V_r$ has a 3D-grid drawing that fits into a
  box of size $r\times 4n\times 4rn$.  The drawing is such that any
  class~$V_i$ is collinear.
\end{lemma}
We include the proof for the reader's convenience.
\begin{proof}
  Let $p$ be the smallest prime with $p\ge 2r - 1$ and set
  $N:=pn/r$. By Bertrand's postulate, $p<4r$ and, hence, $N<4n$.
  For any $0\le
  i\le r-1$, let $V_i=\{(i,t,it) \colon 0\le t<N, t\equiv i^2\pmod p\}$.

  Note that $V_i$ is contained in the line
  $\ell_i=\{(i,0,0)+t(0,1,i) \colon t\in\mathbb R\}$.
  These sets are pairwise disjoint, and each of them has
  precisely $N/p=n/r$ elements.  Connect any two points belonging to
  different $V_i$'s by a straight-line segment.  The resulting drawing
  of $K^r(n)$ fits into a box of size $r\times 4n\times 4rn$, as
  desired.  Pach et al.~\cite[Lemma]{ptt-3dgdg-GD97} showed that no
  two edges of this drawing cross each other.  Moreover, Case~2 of
  their proof implies that, for any~$i$, no edge of this drawing
  crosses a segment between two consecutive vertices of~$V_i$ placed
  along the line~$\ell_i$.  So we can join these vertices by edges
  without adding crossings to the drawing.
\end{proof}

We utilize the construction of Lemma \ref{pttDrawing} to show that $\lva{G} = \pi^1_3(G)$ and that
any graph $G$ admits a drawing which fits in a small 3D integer grid in terms of $\pi^1_3(G)$ and~$n$.

\begin{backInTime}{PiChar0}
\begin{theorem}
\contentThmPiCharZero
\end{theorem}
\end{backInTime}
\begin{proof}
Let $r=\lva G$. The inequality $r\le\pi^1_3(G)$ is obvious.
We now prove $\pi^1_3(G)\le r$.  Let
$V(G)=V'_1\cup\dots\cup V'_r$ be a partition such that
each $G[V'_i]$ is a linear forest. Associate a graph
$K^r(nr)$ with its drawing from Lemma~\ref{pttDrawing}.
Let $V(K^r(rn)) = V_1 \cup\dots\cup V_r$ be the canonical partition
of the set $V(K^r(rn))$. Since $|V'_i|\le |V|=n=|V_i|$ for each $i$,
there exists a map $f \colon V(G)\to V(K^r(rn))$ such that,
for each $i$, $f$ maps adjacent vertices of the linear forest
$G[V'_i]$ into consecutive vertices of the set~$V_i$ placed along a
line~$\ell_i$.

Then the observation at the end of the proof of
Lemma~\ref{pttDrawing} implies that $f(V(G))$ induces a crossing-free
straight-line drawing of~$G$. %
\end{proof}

\begin{backInTime}{thm:PiCharTwo}
  \begin{theorem}
    \contentThmPiCharTwo
  \end{theorem}
\end{backInTime}
\begin{proof}
The bounds $\vt(G)\le\pi^2_3(G)\le\bar\pi^2_3(G)$ are obvious.
The bound  $\bar\pi^2_3(G)\le\vt(G)$ follows from the existence
of a drawing with the specified volume, so we need to prove the last fact.

\newcommand{\pert}{\delta'}

Let $r = \vt(G)$ and let $V_1,\ldots,V_r$ be a partition of the vertex set of $G$ such that
each $G[V_i]$ is a planar graph. As well known, every planar graph admits a plane straight-line
drawing on an $O(n) \times O(n)$ grid~\cite{s-epgg-SODA90,fpp-hdpgg-Comb90}.
Let us fix such a drawing $\delta_i$ for each $G[V_i]$ and place it in
the plane $z=i$. Call an edge $uv$ \emph{horizontal} if both $u$ and $v$
belong to the same $V_i$ for some $i\le r$.
We now have to resolve two problems:
\begin{itemize}
\item
A non-horizontal edge can pass through a vertex of some $\delta_i$;
\item
Two edges that are not both horizontal can cross each other.
\end{itemize}
In order to remove all possible crossings, we replace each $\delta_i$
with its random perturbation $\pert_i$ (still in the same plane $z=i$)
and prove that, with non-zero probability, no crossing occurs.

Specifically, let $s$ and $t$ be parameters that will be chosen later.
Let $T_{a,b,p,q}$ be an affine tranformation of the $(x,y)$-plane
defined by
$$
T_{a,b,p,q}\left(\begin{array}{c}x\\y\end{array}\right)=
\left(\begin{array}{cc}a&-b\\b&a\end{array}\right)\left(\begin{array}{c}x\\y\end{array}\right)+
\left(\begin{array}{c}p\\q\end{array}\right),
$$
where $a$, $b$, $p$, and $q$ are integers such that $0\le p,q<s$, $1\le b<a\le t$, and $a$ and $b$ are coprime.
Note that $T_{a,b,p,q}$ consists of a dilating rotation followed by a shift,
and that it transforms integral points into integral points.
The random drawings $\pert_i$ are obtained by choosing $a$, $b$, $p$, and $q$
at random and applying $T_{a,b,p,q}$ to $\delta_i$
(this is done independently for different $i\le r$).
Note that the resulting drawing occupies a 3D grid of size
$r\times O(tn+s)\times O(tn+s)$.

For each fixed edge $uv$ and vertex $w$ such that $u\in V_i$, $w\in V_j$,
and $v\in V_k$ for some $i<j<k$, let us estimate the probability that
$uv$ passes through $w$ in the drawing.
Conditioned on the positions of $\pert_l$ for all $l\ne j$
and on the choice of the parameters $a$ and $b$ in $T_{a,b,p,q}$ for $\pert_j$,
this probability is clearly at most $1/s^2$. Therefore, this
probability is at most $1/s^2$ also if all $\pert_l$ are chosen at random.
It follows that there is a non-horizontal edge passing through some vertex
with probability at most $mn/s^2$.

Consider now two edges $u_1v_1$ and $u_2v_2$. If there is an $i\le r$
such that $V_i$ contains exactly one of the vertices $u_1$, $v_1$, $u_2$, and $v_2$,
then an argument similar to the above shows that $u_1v_1$ and $u_2v_2$ cross
with probability at most $1/s$. It follows that some edges of this kind
cross each other with probability at most $m^2/s$.

Suppose now that $u_1,u_2\in V_i$ and $v_1,v_2\in V_j$.
Note that shifts cannot resolve the possible crossing of the edges $u_1v_1$ and $u_2v_2$.
Luckily, if we fix $\pert_i$ and ``rotate'' $\pert_j$ by means of $T_{a,b,p,q}$
with random $a,b$ and fixed $p,q$, then $u_1v_1$ and $u_2v_2$ will cross
in at most one case. The probability of this event is bounded by $O(1/t^2)$
because the number of coprime $a$ and $b$ such that $1\le b<a\le t$
is equal to the number of Farey fractions of order $t$, which is known to be
asymptotically $\frac3{\pi^2}t^2+O(t\log t)$ \cite{GKP-book}.
It follows that some edges of this kind
cross each other with probability at most $O(m^2/t^2)$.
Summarizing, we see that the random drawing of $G$ will have a crossing with
probability bounded by
$$
\frac{mn}{s^2} + \frac{m^2}{s} + O\left(\frac{m^2}{t^2}\right).
$$
This probability can be ensured to be strictly smaller than 1
by choosing parameters $s=O(m^2)$ and $t=O(m)$.
We conclude that for such choice of $s$ and $t$ there is at least one crossing-free drawing.
Since $O(tn+s)=O(mn+m^2)=O(m^2)$ (the latter equality being true for $G$ with
no isolated vertex), such a drawing occupies volume $r\times O(m^2)\times O(m^2)$.
\end{proof}

\begin{backInTime}{pi^1_3(K_n)}
\begin{example}
\contentExPiK
\end{example}
\end{backInTime}

\begin{proof}
(a--b)   The lower bound for $\pi^1_3(K_n)$ and the upper bound for $\pi^1_3(K_{p,q})$
follow from Corollary~\ref{PiChi}.
  The upper bound $\pi^1_3(K_n)\le\lceil n/2\rceil$ is given by any
3-dimensional drawing of $K_n$; we can split the vertices in pairs
and draw a line through each pair.
(c)
By Theorem~\ref{thm:PiCharTwo}, $\pi^2_3(K_n)$ is equal to the smallest size $r$ of a partition
$V(G)=V_1\cup\ldots\cup V_r$ such that every $V_i$ induces a planar subgraph of $K_n$,
that is iff every $V_i$ has size at most $4$ (because $K_4$ is planar and $K_5$ is not).
Such a partition exists iff $r\ge \lceil n/4\rceil$.
\end{proof}

\begin{backInTime}{ex:rho12(K_n)}
\begin{example}
\contentExRho
\end{example}
\end{backInTime}

\begin{proof}
  Brief comments on this example:
  (a)~Any line contains at most one of the $\binom{n}{2}$ edges
  of~$K_n$, otherwise the line would contain a triangle.
  (b)~In any drawing that realizes $\rho_3^1(K_{p,q})$, each line
  contains at least one and at most two of the $pq$ edges of~$K_{p,q}$.
\end{proof}

\begin{backInTime}{Kpq}
\begin{example}
\contentExRhoKpq
\end{example}
\end{backInTime}

\begin{proof}
  Indeed, let $S$ be a set of planes underlying a drawing of $K_{p,q}$. Every plane in $S$ contains at most two
  points of either type $p$ or of type $q$, otherwise it would contain
  the non-planar~$K_{3,3}$.  Hence, every plane in~$S$ covers at most
  $2q$ edges. Given $pq$ edges in total, we get~$|S|\ge
  \lceil pq/2q \rceil = \lceil p/2 \rceil$.
  This lower bound is tight.
  Place all points of
  type~$q$ on a line~$\ell$ and introduce $\lceil p/2 \rceil$ distinct
  planes containing~$\ell$.  Line~$\ell$ divides the planes into $2\lceil
  p/2 \rceil\ge p$ half-planes.  Put every point of type $p$ in one of
  these half-planes and connect it to the points on~$\ell$.
\end{proof}

\begin{figure}
  \newcommand{\radius}{2}
  \begin{subfigure}[b]{0.4\linewidth}
    \centering
    \begin{tikzpicture}[]
      \foreach \angle/\i in {1,2,...,6} {
        \coordinate (o\i) at (\i*60+30:\radius);
      }
      \coordinate (o1r)  at ($ (o1) + (0.04,0) $);
      \coordinate (o1b)  at ($ (o1) - (0.04,0) $);
      \coordinate (o2bl) at ($ (o2) + (0.04,0) $);
      \coordinate (o2b)  at ($ (o2) - (0.04,0) $);
      \coordinate (o3bl) at ($ (o3) + (0.04,0) $);
      \coordinate (o3b)  at ($ (o3) - (0.04,0) $);
      \coordinate (o4r)  at ($ (o4) + (0.04,0) $);
      \coordinate (o4b)  at ($ (o4) - (0.04,0) $);
      \coordinate (o5bl) at ($ (o5) - (0.04,0) $);
      \coordinate (o5r)  at ($ (o5) + (0.04,0) $);
      \coordinate (o6bl) at ($ (o6) - (0.04,0) $);
      \coordinate (o6r)  at ($ (o6) + (0.04,0) $);
      \draw[black] (o2bl) -- (o3bl) -- (o6bl) -- (o5bl) -- (o2bl)
        (o3bl) -- (o5bl)
        (o2bl) -- (o6bl);
      \draw[red] (o1r) -- (o4r) -- (o5r) -- (o6r) -- (o1r)
        (o1r) -- (o5r)
        (o4r) -- (o6r);
      \draw[blue] (o1b) -- (o4b) -- (o3b) -- (o2b) -- (o1b)
        (o2b) -- (o4b)
        (o3b) -- (o1b);
      \foreach \angle/\i in {1,2,...,6} {
        \fill(o\i) circle (2pt);
      }
    \end{tikzpicture}
    \caption{$c(K_6,K_4)=3$.}
    \label{fig:cK6K4}
  \end{subfigure}
  \hfill
  \begin{subfigure}[b]{0.4\linewidth}
    \centering
    \begin{tikzpicture}[]
      \foreach \angle/\i in {1,2,...,7} {
        \coordinate (o\i) at (\i*51.43+37:\radius);
      }

      \draw[green!70!black] (o2) -- (o3) -- (o7) -- cycle;
      \draw[green!70!black] (o3) -- (o5) -- (o6) -- cycle;
      \draw[red] (o4) -- (o6) -- (o7) -- cycle;
      \draw[blue] (o1) -- (o2) -- (o6) -- cycle;
      \draw[blue] (o1) -- (o3) -- (o4) -- cycle;
      \draw[black] (o2) -- (o4) -- (o5) -- cycle;
      \draw[black] (o1) -- (o5) -- (o7) -- cycle;

      \foreach \angle/\i in {1,2,...,7} {
        \fill(o\i) circle (2pt);
      }
    \end{tikzpicture}
    \caption{$c(K_7,K_3)=7$.}
    \label{fig:cK7K3}
  \end{subfigure}
  \caption{Combinatorial bounds for the numbers of $K_3$'s and $K_4$'s
    needed to cover $K_7$ and $K_6$, respectively.}
\end{figure}
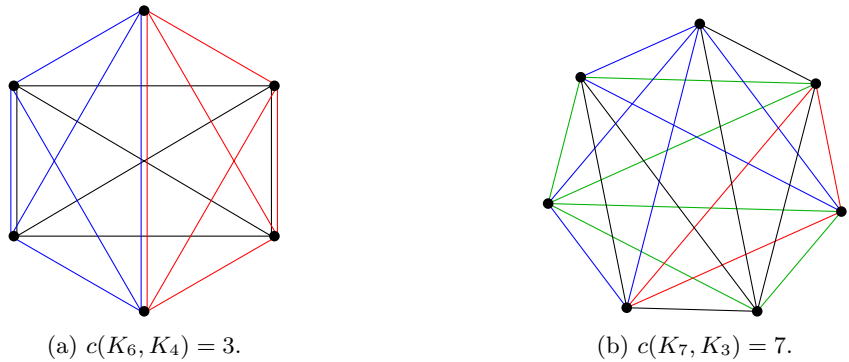
\begin{backInTime}{thm:KnLowerUpper}
  \begin{theorem}
    \contentThmKnLowerUpper
  \end{theorem}
\end{backInTime}
\begin{proof}
  For the lower bound, take a drawing of $K_n$ with a geometric cover $\mathcal
  L$ using $\rho^2_3(K_n)$ planes.
This geometric cover induces a combinatorial cover
  $\mathcal K_{\mathcal L}$ of $K_n$ by $\rho^2_3(K_n)$ copies of~$K_{\leq 4}$.
It follows that $\rho^2_3(K_n)\ge c(K_n,K_4)$.

  For the upper bound, let $d \colon V(K_n)\to\mathbb R^3$ be an
  arbitrary drawing of~$K_n$ in 3-space (with non-crossing edges).
  Since, for any $K_3$-subgraph of $K_n$, its image $d(K_3)$ is
  contained in a plane, it is clear that $\rho^2_3(K_n)\le
  c(K_n,K_3)$.

  Now we show lower and upper bounds for $c(K_n, K_3)$ and $c(K_n, K_4)$
and determine their asymptotics.
  Since the graph $K_n$ has $\binom{n}{2}=n(n-1)/2$ edges
  and each copy of the graph $K_k$ has $\binom{k}{2}=k(k-1)/2$ edges,
  we see that $c(K_n,K_k)\ge \frac{n(n-1)}{k(k-1)}$,
  in particular we get $c(K_n,K_4)\ge n^2/12-n/12$ for $k=4$.
  This lower bound
  is attained provided there exists a Steiner system
  $S(2,k,n)$, so in this case $c(K_n,K_k)=|S(2,k,n)|=\frac{n(n-1)}{k(k-1)}$.
Hanani~\cite{h-ecbibd-AMS61}
\savespace{(see also~\cite[Theorem 2.1]{rr-ssss-EJC10})}
showed that a Steiner system
$S(2,4,n)$ exists iff $n\equiv 1\pmod{12}$ or $n\equiv 4\pmod{12}$.
This implies that $c(K_n, K_4)=n^2/12+\Theta(n)$ for any~$n$.

  In 1847 Kirkman showed (see, for instance,~\cite[p.113]{Bol1}) that a Steiner
  system $S(2,3,n)$ exists iff $n\equiv 1\pmod 6$ or $n\equiv 3\pmod 6$ (for instance,
  Fig.~\ref{fig:cK7K3} shows that $c(K_7,K_3)=7$).  Hence, for any~$n$,
  we have $c(K_n, K_3)=n^2/6+O(n)$.
\end{proof}

\begin{backInTime}{lem:KnCovTriangle}
  \begin{lemma}
    \contentLemKnCovTriangle
  \end{lemma}
\end{backInTime}
\begin{proof}
  In any planar drawing of $K_4$, the four
  vertices cannot be drawn as vertices of a convex quadrilateral for else
  its diagonals would intersect.
\end{proof}

\begin{backInTime}{lem:KnCovCommon}
  \begin{lemma}
    \contentLemKnCovCommon
  \end{lemma}
\end{backInTime}
\begin{proof}
  Indeed, assume the converse: both graphs~$K_\ell$ and
  $K_{\ell'}$ contain the same copy~$K_3'$ of~$K_3$.  Since the
  triangle $d(K_3')$ cannot be collinear, both
  $d(K_\ell)$ and $d(K_{\ell'})$ lie in the plane~$\ell''$ spanned
  by the set~$d(K_3')$. But then the plane~$\ell''$ contains five
  points of~$d(K_n)$, which is impossible.
\end{proof}

\begin{lemma}
  \label{lem:RhoKn+1}
  \begin{enumerate}[a)]
    \item $\rho^2_3(K_{n+1}) \leq \rho^2_3(K_n) + \lceil n/2 \rceil$.
    \item $\rho^2_3(K_{n+1}) \leq \rho^2_3(K_n) + \lceil (n-3)/2 \rceil$ if
      there is a geometric cover of a drawing $d$ of $K_n$ realizing the value of
      $\rho^2_3(K_n)$, where one of the covering planes contains exactly three
      vertices.
  \end{enumerate}
\end{lemma}
\begin{proof}
  \begin{enumerate}[a)]
    \item Since each drawing of the graph $K_n$ can be
      extended to a drawing of the graph $K_{n+1}$ by adding $n$
      segments with a common endpoint $d(v_{n+1})$, which can be covered by
      $\lceil n/2\rceil$ planes, we see that $\rho^2_3(K_{n+1})\le\rho^2_3(K_n)+\lceil
      n/2\rceil$.
    \item
      Let $\ell$ by the covering plane that contains exactly three vertices $d(v)$,
      $d(v')$ and $d(v'')$ of $d(K_n)$.
      If we extend $d$ to a drawing of the graph $K_{n+1}$ by adding
      the endpoint $d(v_{n+1})$ inside of the triangle with the vertices $d(v)$, $d(v')$ and $d(v'')$
      then it suffices to cover by additional planes only $n-3$
      segments with a common endpoint $d(v_{n+1})$, connecting it with vertices of
      $d(V(K_n))\setminus\ell$. \qedhere
  \end{enumerate}
\end{proof}

\begin{backInTime}{thm:rhoCompleteGraphs}
  \begin{theorem}
    \contentThmRhoCompleteGraphs
  \end{theorem}
\end{backInTime}
\begin{proof}
$n=5$. By Lemma~\ref{lem:RhoKn+1}(a), $\rho^2_3(K_5)\le \rho^2_3(K_4)+2=3$.
To obtain a lower bound, remark that $\rho^2_3(K_5)\ge c(K_5, K_4)=3$.
To prove the last equality remark that although a graph $K_5$ has
$10<2\cdot 6$ edges, it cannot be covered by two copies $K_4'$ and $K_4''$ of a graph $K_4$,
because in this case they should have at least three common vertices, so their intersection should have at
least three common edges, but $12-3<10$. From the other side,
each two different copies of $K_4$ cover all edges of $K_5$ but one, so $c(K_5, K_4)=3$.

$n=6$. This case is treated in the main part of the paper
(page~\pageref{thm:rhoCompleteGraphs}).

$n=7$. Since in the cover of $d(K_6)$ in Fig.~\ref{fig:K6-photo}
by $4$ planes, one of the covering planes contains exactly three vertices $d(v)$,
by Lemma~\ref{lem:RhoKn+1}(b), we obtain $\rho^2_3(K_7)\le 4+\lceil (6-3)/2\rceil=6$,
Now we show that $\rho^2_3(K_7)\ge 6$. Assume that $r=\rho^2_3(K_7)<6$.
Consider a combinatorial cover $\mathcal K_{\mathcal L}$ of $K_7$ by its complete planar subgraphs
corresponding to a geometric cover $\mathcal L$ of its drawing by $r$ planes.
Count number $s=|\{(v,K):v\in V(K_7),K\in\mathcal K_{\mathcal L}\}|$. Since
$|\mathcal K_{\mathcal L}|=r$, and each graph $K\in\mathcal K_{\mathcal L}$ has at most $4$
vertices, $s\le 5\cdot 4=20<21=7\cdot 3$. Therefore there exists a vertex $v_0\in V(K_7)$ covered
by at most two members of the cover $\mathcal K_{\mathcal L}$. Since each element $K$ of
cover $\mathcal K_{\mathcal L}$ covers at most three vertices incident to $v_0$ (and exactly three vertices
only if $K$ is a copy of graph $K_4$) and, in graph $K_7$, there are $6$ edges
incident to vertex $v_0$, we see that vertex $v_0$ belongs to exactly two members
$K^1_4$ and $K^2_4$ of the cover $\mathcal K_{\mathcal L}$, and each of these members is a copy of
graph $K_4$. Moreover, $v_0$ is the unique common vertex of the graphs $K^1_4$ and $K^2_4$.
For each $i$ let $V_i=V(K^i_4)\setminus\{v_0\}$.
Let $\mathcal K'=\mathcal K_{\mathcal L}\setminus\{K^1_4,K^2_4\}$.
Since $|\mathcal K'|=r-2$ and $r-2<4=\rho^2_3(K_6)$, there exists an edge $(v,v')$
of the graph $K_7-v_0$ which is not covered by the family $\mathcal K'$.
Since a set $\{(u,u'):u\in V_1, u'\in V_2\}$ of edges is covered by the family $\mathcal K'$,
there exists an index $i$ with $V_0=\{v,v'\}\subset V_i$. Let $j\ne i$ be the other index.
We have $|V_0\cap V(K)|\le 1$
for each $K\in\mathcal K'$. Since $|\mathcal K'|=r-2\le 3$, there exists a vertex $v_1\in V_0$
which belongs to at most one set $K\in \mathcal K'$. In fact, such a set $K$ exists, because
in in opposite case no edge $(v_1,w)$ for $w\in V_j$ is covered by $\mathcal K_{\mathcal L}$.
Since both $K^j_4$ and $K$ are members of the cover $\mathcal K_{\mathcal L}$, by
Lemma~\ref{lem:KnCovCommon}, there exists a vertex $w\in V_j\setminus V(K)$. Then
an edge $(v_1,w)$ is not covered by $\mathcal K_{\mathcal L}$, a contradiction.

$n=8$.  Clearly, $\rho^2_3(K_8)\ge \rho^2_3(K_7)=6$.  Put a point
$d(u_7)$ inside of a triangle $d(u_2)d(u_5)d(u_6)$ and a point $d(u_8)$ inside
of a triangle $d(u_3)d(u_5)d(u_6)$ in the drawing of $d(K_6)$ in
Fig.~\ref{fig:K6-photo} symmetrically with respect to the axis
$d(u_1)d(u_4)$.  Then, to cover all edges of the drawing~$d(K_8)$, it
suffices to add the four planes of Fig.~\ref{fig:K6-photo}, an
additional plane spanned by triangle $d(u_1)d(u_7)d(u_8)$ and lines
spanned by segments $d(u_3)d(u_7)$ and $d(u_2)d(u_8)$.  Therefore,
$\rho^2_3(K_8)\le 7$.

$n=9$. $\rho^2_3(K_9)\ge c(K_9,K_4)>6$. We prove the last inequality.
Since a graph $K_9$ has $36=6\cdot 6$ edges, each cover of $K_6$ by
six copies of $K_4$ generates a Steiner system $S(2,4,9)$. The absence of such a system
follows from the result of Hanani mentioned earlier, but we give a direct proof.
Indeed, assume that $c(K_9,K_4)\le 6$. Since degree of each vertex $v$ of $K_9$ is $8$,
$v$ belongs to at least $3$ copies of $K_4$ from the cover $\mathcal K'$. Then
$s=|\{(v,K):v\in V(K_8),K\in\mathcal K'|\ge 9\cdot 3=27$.
But since each member $K$ of the cover $\mathcal K'$ contains $4$ vertices of $V(K_9)$,
we have $s\le 6\cdot 4=24$, a contradiction.
\end{proof}

\section{The Affine Cover Numbers of Planar Graphs: Proofs}
\label{sec:plan-appendix}

\begin{backInTime}{rem:PlanarLvaThree}
  \begin{example}
    \contentRemPlanarLvaThree
  \end{example}
\end{backInTime}
\begin{proof}
Indeed, in the picture it is easy to see that $2 \le
\lva G \le 3$.  In order to show that $\lva G>2$, assume that the
vertex set of $G$ is colored black and white where each monochromatic component induces a linear forest.
Without loss of generality, we may
assume that the central vertex is white.  Since the central vertex
cannot have more than two white neighbors, at most two other vertices
are white.  Note that the neighbors of the central vertex form a
square in Fig.~\ref{fig:three-lines} and that each side of the square
contains a cycle.  Hence, none of the sides of the square can be
monochromatic; it must contain at least one white vertex.  Therefore,
the boundary of the square contains exactly two white vertices, which
must be placed in opposite corners. If the white vertices are placed
in the top left and the bottom right corners, then the two other
corners, which are black, have three black neighbors.  If the white
vertices are placed in the top right and the bottom left corners, then
the three white vertices induce a cycle.  In both cases, we have a
contradiction.
\end{proof}

\begin{figure}[tb]
  \centering
  \includegraphics[page=2,scale=1.2]{trdr}
  \hfill\includegraphics[page=1,scale=1.2]{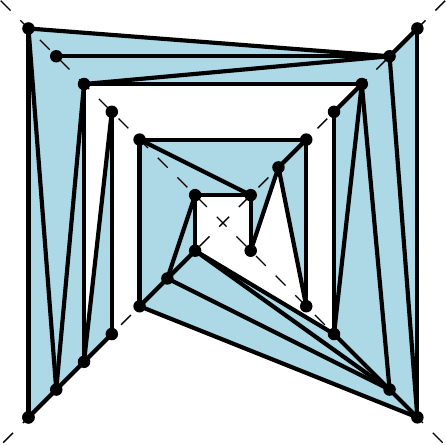}
  \caption{Proof of Theorem~\ref{PiTrDr}}
  \label{fig:trdr}
\end{figure}

\begin{backInTime}{thm:segm-vs-rho12}
\begin{theorem}
\contentThmSegmRho
\end{theorem}
\end{backInTime}
\begin{proof}
  Call a vertex $v$ of the graph $G$ \emph{branching} if its degree is at least three.
A path between vertices of the graph $G$ will be called \emph{straight} if
it has no branching vertices other than its endpoints.

Reduce the graph $G$ to a graph $G'$ as follows. The set of vertices of $G'$
is the set of branching vertices in $G$, and two vertices are adjacent in $G'$
if they are connected by a straight path in $G$.
Being a planar graph, $G'$ has a straight-line drawing.

If $G'$ is empty then $G$ is a path or a cycle, and $\segm(G)\le 3$.
So, assume that $G$ has a branching vertex. Since $G$ is connected,
every vertex in it is connected to a branching vertex by a straight path (possibly of zero length).
Thus, we can construct a straight-line drawing of $G$ from a straight-line drawing of $G'$ as follows.
Note that an edge $e$ of $G'$ corresponds to a bond of paths in $G$ connecting the
incident vertices of $e$.
We draw one path in the bond on the segment $e$, and each other is drawn as a pair of two segments
that are close to $e$.
A branching vertex $v$ in $G$ can be connected by a straight path to a degree 1 vertex
(which disappears in $G'$). We restore each such path by drawing it as a small segment.
Moreover, $v$ can belong to a cycle whose all vertices except $v$ have degree 2.
We draw each such cycle as a small triangle.

Note that the segments incident to
the branching vertex $v$ are split into three parts: $b_v$ segments which belong to the edge bonds,
$l_v$ segments going to leaves, and $t_v$ segments that are sides of the small triangles.
Unless $l_v=0$ and $t_v=2$, we can ensure that the last $l_v+t_v$ segments are drawn in at most $\lceil (l_v+t_v)/2 \rceil$
all crossing at the point $v$.
Therefore, the vertex $v$ is incident to at most $b_v+\lceil (l_v+t_v)/2 \rceil$ segments.
This holds true even if $l_v=0$ and $t_v=2$; we need to use the fact that in this case $b_v\ne0$.
We will say that these segments are \emph{related} to $v$. We also relate to $v$
the opposite sides of the corresponding small triangles; there are
$t_v/2$ of them. Thus, the vertex $v$ has at most
$b_v+\lceil (l_v+t_v)/2 \rceil+t_v/2\le b_v+\lceil (l_v+t_v)/2 \rceil+\lfloor (l_v+t_v)/2 \rfloor=b_v+l_v+t_v=\deg v$
segments. Since every segment of the constructed drawing of $G$ is related to some vertex,
the total number of the segments is bounded by
$\sum_{v\in V(G), \deg v\ge 3} \deg v\le
4\sum_{v\in V(G), \deg v\ge 3}\left\lceil\frac{\deg v}{2} \right\rceil\left(\left\lceil\frac{\deg v}2\right\rceil-1\right)
\le 4\rho^1_3(G)(\rho^1_3-1)$.
The last estimate is the first inequality in the proof of  Lemma~\ref{Ess}(b).
\end{proof}

\label{sec:edges-appendix}
\begin{figure}
  \begin{minipage}[b]{.95\textwidth}
    \begin{subfigure}[b]{.48\textwidth}
      \centering
      \includegraphics{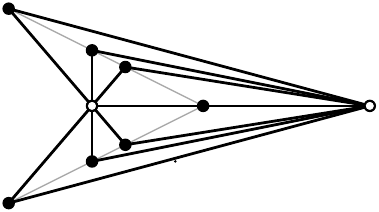}
      \caption{$\rho^1_2(K_{2,q})\le \lceil (3n-7)/2\rceil$.}
       \label{fig:K2n-2upper}
    \end{subfigure}
    \hfill\hfill
    \begin{subfigure}[b]{.48\textwidth}
      \centering
      \includegraphics{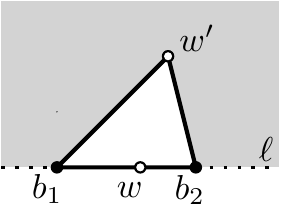}
      \caption{$\rho^1_2(K_{2,q})\ge\lceil (3n-7)/2\rceil$.}
       \label{fig:K2n-2lower}
    \end{subfigure}
    \caption{Establishing lower and upper bounds for~$\rho^1_2(K_{2,q})$.}
  \end{minipage}
\end{figure}

\begin{backInTime}{K2n-2}
\begin{example}
\contentExKTwoNMinusTwo
\end{example}
\end{backInTime}

\begin{proof}
The former equality is obvious. We have to prove for $G=K_{2,q}$ that $\rho^1_2(G)=\lceil (3n-7)/2\rceil$.
Fig.~\ref{fig:K2n-2upper} shows that
$\rho^1_2(G)\le \lceil (n-3)/2\rceil+n-2=\lceil (3n-7)/2\rceil$.
It remains to show the lower bound $\rho^1_2(G)\ge \lceil (3n-7)/2\rceil$.
Suppose that our bipartition is defined by $2$ white vertices
and $q$ black vertices. Associate the graph $G$ with its plane drawing.
If there exist no line containing one white vertex and two black vertices of the graph $G$
then we need $m=2n-4\ge \lceil (3n-7)/2\rceil$ lines to cover all edges of the graph $G$.
Assume from now that
there exists a line $\ell$ containing a white vertex $w$ and two black vertices $b_1$ and $b_2$ of the graph
$G$. Then the vertex $w$ lies on the line $\ell$ between the vertices $b_1$ and $b_2$.
Let $w'$ be the other white vertex.
Since the point $w$ sees all black points, no one of them can be placed inside the shaded area,
see Fig.~\ref{fig:K2n-2lower}. Then the point $w'$ cannot be an interior point of a segment between
two black points.

Thus all lines which cover at least two edges of the graph $G$ go through the point $w$, at
most one of these lines go through the point $w'$ and the remaining lines can cover only
the edges incident to the vertex $w$. Now let $\mathcal L$ be a family of
lines such that $|\mathcal L|=\rho^1_2(G)$ and each edge of the graph $G$ belongs to some
line $\ell\in\mathcal L$. Let $\mathcal L_w=\{\ell\in\mathcal L:w\in\ell\}$,
$\mathcal L_{w'}=\{\ell\in\mathcal L:w'\in\ell\}$, and $\mathcal L_{w,w'}=\mathcal L_w\cap
\mathcal L_{w'}$. Clearly, $|\mathcal L_w|\ge\rho^1_2(K_{1,n-2})=\lceil (n-2)/2\rceil$.
By the above, $|\mathcal L_{w'}|=n-2$ and $|\mathcal L_{w,w'}|\le 1$.
If the set $\mathcal L_{w,w'}$ is empty, then
$|\mathcal L|=|\mathcal L_{w}|+|\mathcal L_{w'}|\ge\lceil (n-2)/2\rceil+n-2=\lceil (3n-6)/2\rceil$.
If  $\mathcal L_{w,w'}=\{\ell_0\}$ then since the line $\ell_0$ covers exactly
edges $(w,b)$ and $(w',b)$ for some black vertex $b$,
$|\mathcal L_w\setminus \{\ell_0\}|\ge \rho^1_2(K_{1,n-3})=\lceil (n-3)/2\rceil$, and
$|\mathcal L|=|\mathcal L_w\setminus \{\ell_0\}|+
|\mathcal L_{w'}|\ge\lceil (n-3)/2\rceil+n-2=\lceil (3n-7)/2\rceil$.
In both cases. $\rho^1_2(G)\ge \lceil (3n-7)/2\rceil$.
\end{proof}

\begin{figure}[tb]
  \centering
  \includegraphics[scale=0.9]{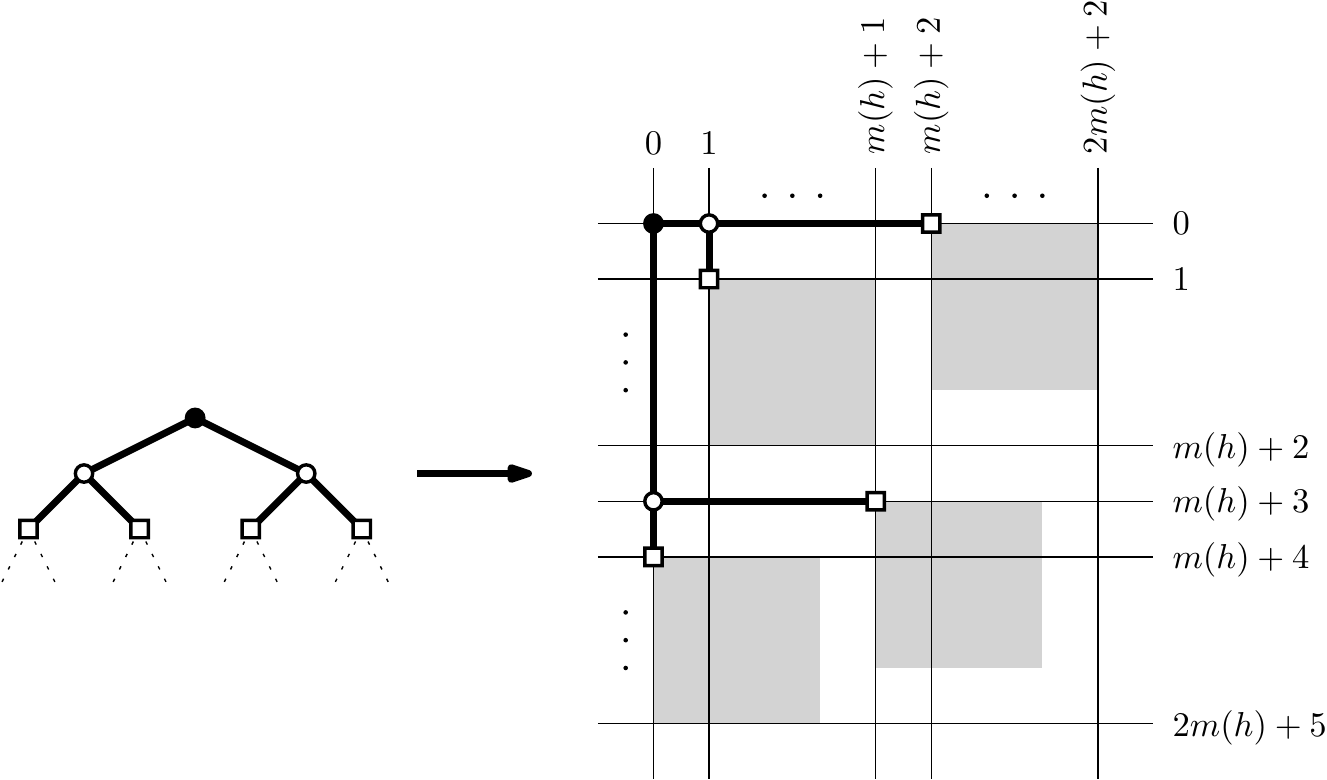}
  \caption{Drawing of a complete binary tree of height $h+2$ on a grid of size
  $(2m(h)+2)\times (2m(h)+5)$.}
  \label{fig:bintree-tight-rho}
\end{figure}
\begin{example}\label{prop:CompleteBinTreeRho12}
 If $G$ is the complete binary tree of height $h \ge 1$ consisting of $n$ nodes, then
$\rho^1_2(G)>\sqrt{n-3}$. On the other hand,
$\rho^1_2(G)\le (6/{\sqrt{2}})\sqrt{n+1}-5$ if $h$ is even
and $\rho^1_2(G)\le 4\sqrt{n+1}-5$ otherwise.
\end{example}
\begin{proof}
Indeed, since $G$ has $(n-3)/2$ vertices with degree $3$,
$\rho^1_2(G)>\sqrt{n-3}$ by Lemma~\ref{Ess}(a).

To obtain an upper bound, let $m(2)=2$, $m(3)=4$, and $m(h+2)=2m(h)+4$
for $h\ge 2$.
We prove by induction that we can draw
the tree of height $h\ge2$ on an $m(h)\times(m(h)+1)$-grid
with the root placed at the top left corner
(see also \cite[Fig.~3(a)]{gr-sldbt-JGAA04}).
For the induction basis note that this is true for $h=2$ and $h=3$.
We can also draw the tree of height $0$ on a $0\times 0$-grid and
the tree of height $1$ on a $1\times 1$-grid, respectively,
with the root placed at the top left corner.
As the induction step, we use a $m(h+2)\times(m(h+2)+1)$ grid
and place the root of the tree of height $h+2$ on the top left grid point~$(0,0)$,
place its child nodes on points~$(1,0)$
and~$(0,m(h)+3)$, and place its grandchild nodes on points~$(m(h)+2,0)$, $(1,1)$, $(m(h)+1, m(h)+3)$ and
$(0,m(h)+4)$ (see Fig.~\ref{fig:bintree-tight-rho}).
For each grandchild~$v$, consider the intersection of $m(h)$ gridlines to the
right and the $m(h)+1$ gridlines to the
bottom of~$v$ (including the grid lines containing~$v$).
We reserve this part of the grid to the subtree of~$v$.
Note that every grid point is reserved to at most one grandchild due to the placement of the
grandchild nodes. Since the subtree of a grandchild is a complete binary tree of height $h$,
by induction, we can draw each of them in the reserved part of size $m(h)\times (m(h)+1)$
with the roots placed at the top left corners.
Thus all edges of the tree of height $h$ belong to $2m(h)+1$ grid lines.
We can easily show that $m(h)=6\cdot 2^{(h-2)/2}-4$ if $h\ge 2$ is even
and $m(h)=8\cdot 2^{(h-3)/2}-4$ if $h\ge 3$ is odd.
Since $n=2^{h+1}-1$, we get the claimed upper bounds.
\end{proof}

\begin{figure}[b]
  \centering
  \includegraphics[width=\textwidth]{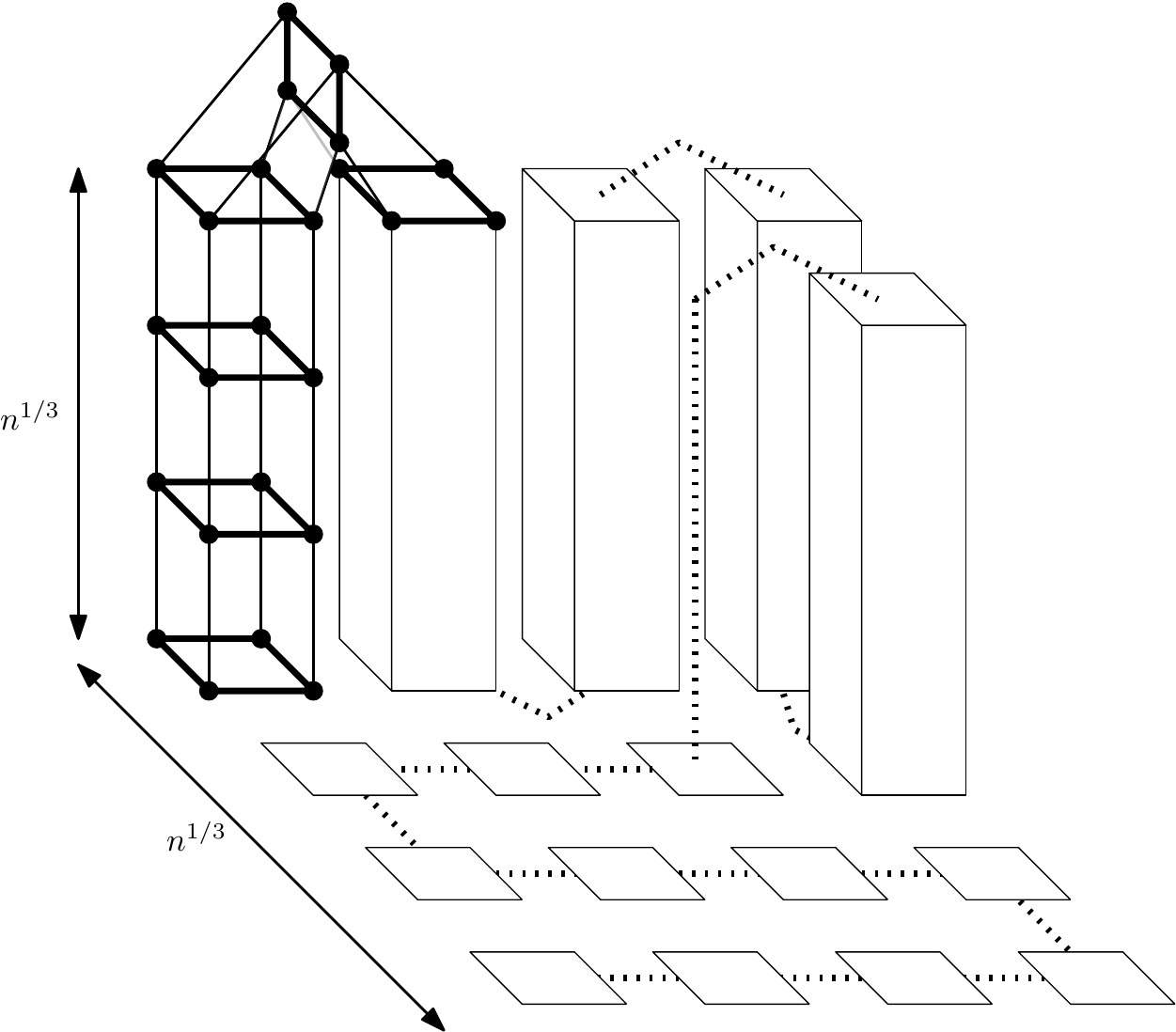}
  \caption{The graph $C_4\times P_k$ drawn into a 3D grid of linear volume
    on $O(n^{2/3})$ lines (the proof of Theorem~\ref{thm:nested-triangles}(b)).}
  \label{fig:3D-grid}
\end{figure}

\section{Conclusion and Open Problems}
\label{sec:open-problems}

Apart from the many open problems that we have mentioned throughout
the paper, we suggest to study a topological version of~$\rho^l_d$,
say $\tau^l_d$, where there are $\tau^l_d(G)$ $l$-dimensional planes
such that each edge of $G$ is contained in one such plane. It seems
that $\tau^1_d(G) = \rho^1_d(G)$, and for $d \geq 3$, $\tau^3_d(G)=
1$. So, the only interesting such parameter is $\tau^2_3(G)$. This
would relate more closely to book embeddings.

Can we bound $\rho^2_3$ (and $\pi^2_3$) for 1-planar graphs
or RAC graphs? %
Is $\rho^2_3$ bounded by a constant or linear function for bounded-degree
graphs?
Are there tighter bounds for $\rho^2_3(K_n)$ than those in Theorem~\ref{thm:KnLowerUpper}?
Find $\lim_{n\to\infty} \rho^2_3(K_n)/n^2$ if it exists.

Given that it is NP-complete to decide whether the vertex arboricity
of a maximal planar graph is at most $2$~\cite{hs-va-SDM89},
how hard it is to test whether $\pi^1_3(G)=2$ for a planar graph~$G$?
\end{document}